\documentclass[11pt,oneside]{amsart}
\title[Asymptotically-informed NNs for implied volatility computation]{Asymptotically-informed neural networks for Black-Scholes implied volatility computation}
\date{}

\author{Samira Amiriyan}
\address[Samira Amiriyan]{Institute for Financial and Actuarial Mathematics\\
Department of Mathematical Sciences\\ 
The University of Liverpool\\
Mathematical Sciences Building\\
Peach Street. \\
Liverpool, L69 7ZL. United Kingdom}
\email{Samira.Amiriyan@liverpool.ac.uk}

\author{Youness Boutaib}
\address[Youness Boutaib]{Institute for Financial and Actuarial Mathematics\\
Department of Mathematical Sciences\\ 
The University of Liverpool\\
Mathematical Sciences Building\\
Peach Street. \\
Liverpool, L69 7ZL. United Kingdom}
\email{Youness.Boutaib@liverpool.ac.uk}

\usepackage{amssymb} 
\usepackage{amsmath} 
\usepackage{amsthm} 
\usepackage[margin=1in]{geometry}  
\usepackage[usenames,dvipsnames]{xcolor} 
\usepackage{float} 
\usepackage{graphicx} 
\usepackage{hyperref}
\usepackage{subcaption} 
\usepackage{adjustbox} 

\usepackage{multirow}
\usepackage{tabularx} 

\usepackage{booktabs}
\usepackage{makecell}

\usepackage[english]{babel}

\newtheorem{theo}{Theorem}[section]
\newtheorem{prop}[theo]{Proposition}

\newcommand{\lbf}{\mathbf{l}}
\newcommand{\ubf}{\mathbf{u}}

\newcommand{\Nbb}{\mathbb{N}}

\newcommand{\Rbb}{\mathbb{R}}

\newcommand{\crm}{\mathrm{c}}
\newcommand{\drm}{\mathrm{d}}
\newcommand{\lrm}{\mathrm{l}}
\newcommand{\urm}{\mathrm{u}}

\newcommand{\inti}[2]{[\![ #1, #2 ]\!]}

\newcommand{\rom}[1]{\uppercase\expandafter{\romannumeral #1\relax}}

\DeclareMathOperator{\erfcx}{erfcx}

\usepackage[english]{babel}
\keywords{Black-Scholes, implied volatility, neural networks, mixture-of-experts architectures, financial machine learning}
\subjclass{Primary: 91G60 ; Secondary: 68T07, 65D15}
\begin{document}	
\maketitle
	\begin{abstract}
	The computation of Black-Scholes implied volatility is a fundamental task in quantitative finance, underpinning option valuation, model calibration and risk management. Although implied volatility is routinely used in practice, the inversion of the Black-Scholes pricing formula remains a challenging numerical problem, particularly in asymptotic regimes corresponding to extreme option prices, strikes or maturities, where the inverse map becomes highly sensitive to perturbations of the price. In this paper, we introduce a new family of asymptotically-informed neural-network architectures for implied-volatility computation. Exploiting the distinct behaviours of the Black-Scholes pricing function in different volatility regimes, we propose a family of architectures that learn a trainable partition of the price--log-moneyness domain through a system of gating functions and combines specialised local approximations of the implied-volatility function within each region. Extensive numerical experiments demonstrate that the proposed models consistently outperform standard feed-forward neural networks across a wide range of parameter domains, often by several orders of magnitude in relative accuracy while maintaining excellent generalisation properties. Furthermore, the neural-network outputs provide highly accurate initial guesses for a third-order Householder scheme, allowing near machine-precision implied-volatility computations after only two refinement iterations.
	\end{abstract}
\section{Introduction}
	Since its introduction by Black and Scholes \cite{BS} and its subsequent extension by Merton \cite{Merton}, the Black-Scholes model has become one of the cornerstones of modern quantitative finance. Given a spot price $S_0$, a strike $K$, a maturity $T$, a risk-free rate $r$ and a volatility parameter $\sigma$, the model provides a closed-form formula for the price $C(S_0,K,T,r,\sigma)$ of a European option. Among these parameters, volatility occupies a unique position. Indeed, while quantities such as the strike, maturity or interest rate are directly observable, volatility is not. Instead, it is inferred from market prices through the inversion of the Black-Scholes formula. More precisely, because the Black-Scholes call price is a strictly increasing function of $\sigma$, every arbitrage-free option price corresponds to a unique volatility level. This quantity, known as the \emph{implied volatility}, is defined as the unique solution $\sigma$ of 
	\[ C(S_0,K,T,r,\sigma)=C, \] 
	where $C$ denotes the observed market price of the option. 
	
	Although the Black-Scholes model assumes a constant volatility parameter, empirical observations reveal that implied volatilities vary systematically across strikes and maturities. This phenomenon gives rise to the well-known implied volatility surface, whose shape plays a central role in derivative pricing, risk management and model calibration. As a consequence, market practitioners often quote options in terms of implied volatilities rather than prices.
	
	The practical importance of implied volatility extends well beyond the Black-Scholes framework itself. First, implied volatilities provide a common scale on which options with different strikes, maturities and underlying prices can be compared. Second, portfolio hedging strategies rely heavily on volatility sensitivities such as Vega, whose computation requires accurate implied-volatility estimates. Third, many modern financial models are calibrated not directly to prices but to implied-volatility surfaces. A prominent example is Dupire's local-volatility framework \cite{Dupire}, where local volatility is recovered from derivatives of the implied-volatility surface. Consequently, inaccuracies in implied-volatility computations may propagate through an entire pricing, hedging or calibration pipeline.
	
	This issue becomes particularly acute in the \emph{asymptotic regimes} of very small or very large option prices. 
	Equivalently, these are regions where the option is either deeply in-the-money or deeply out-of-the-money, or where the time-to-maturity is very small or very large.
	In such situations, the implied-volatility function becomes highly sensitive to perturbations of the price. 
	As a consequence, even small pricing errors may induce large implied-volatility errors, potentially leading to misleading assessments of relative value, incorrect hedging decisions or inaccurate calibration results. From a theoretical perspective, the behaviour of implied volatility in these problematic asymptotic regimes has been extensively studied \cite{Tehranchi, Tehranchi2, Gulisashvili, AL}. However, many asymptotic formulae become accurate only for strike and maturity values that are far beyond those observed in practice \cite{AL}. Consequently, while mathematically insightful, these formulas are often insufficient on their own for the construction of robust implied-volatility solvers operating over realistic ranges of parameter values.
	
	The issue of accurately computing the implied volatility in the asymptotic regimes is not merely academic and remains relevant even when the Black-Scholes model itself is not used as the final pricing model. In practice, sophisticated stochastic or rough volatility models are often employed to generate option prices, either through semi-closed formulas or Monte Carlo schemes. These prices are then converted into implied volatilities before any meaningful comparison with market data is carried out. Consequently, a highly accurate pricing methodology becomes of limited practical value if the final implied-volatility inversion introduces significant errors. In particular, two competing parameter sets of a stochastic-volatility model may produce almost identical prices while differing substantially in the tails of the implied-volatility surface. Thus, a small error in the inversion procedure may distort the calibration process itself. Moreover, optimisation algorithms used for calibration may explore regions of the parameter space corresponding to implied volatilities that are far from those directly observed in market data. Therefore, the calibration procedure requires accurate implied-volatility computations not only in typical market conditions, but also in extreme regimes visited during the optimisation process. In summary, and regarding calibration, inaccuracies in the implied volatility in asymptotic regions may lead to incorrect parameter estimates, poor out-of-sample performance or misleading conclusions regarding the relative quality of competing models.
	
	Finally, accurate implied-volatility computations are essential for risk management. The low-volatility regime is particularly sensitive because the option price varies rapidly with respect to the volatility parameter in relative terms. Consequently, small errors in implied volatility may lead to significant errors in volatility sensitivities and therefore in hedging strategies. In extreme cases, numerical inaccuracies may even create the illusion of relative-value or arbitrage opportunities that are not actually present in the market.
	
	Despite the importance of implied volatility, no explicit closed-form inverse of the Black-Scholes formula is known. As a result, implied-volatility computation has generated a large literature spanning numerical root-finding methods, analytical approximations and, more recently, machine-learning approaches.
	
	A first class of methods formulates implied-volatility computation as a root-finding problem. Early work in \cite{MK} proposed the Newton-Raphson method, while Brent's algorithm \cite{Brent} subsequently emerged as a popular alternative due to its robustness. The performance of such methods depends critically on the quality of the initial guess and may deteriorate in extreme parameter regimes.
	
	A second class of methods seeks analytical approximations of the implied-volatility function itself. Representative examples include \cite{BS2, BCS, Chance, CM, CN, Li}. Among these approaches, the method of J\"ackel \cite{Jackel} is generally regarded as one of the most accurate and robust available techniques. A key insight of \cite{Jackel} is that the Black-Scholes price exhibits qualitatively different behaviours depending on the volatility regime. By partitioning the parameter space and constructing specialised approximations in each region, J\"ackel obtains highly accurate initial guesses which are then refined through a high-order root-finding scheme. Inspired by this idea, a global approximation based on Chebyshev interpolation was later proposed in \cite{GHMP}, trading some accuracy for computational simplicity.
	
	In parallel, the last decade has witnessed a rapid expansion of machine-learning methods in quantitative finance. Neural networks have been employed in option pricing \cite{GG, HLP, FG}, volatility-surface modelling \cite{ATV, VC, NZW}, model calibration \cite{Hernandez, BHMST, HMT, CKT}, and the numerical solution of financial partial differential equations \cite{SS, SS2}. Comprehensive surveys such as \cite{RW} document the significant impact that deep-learning techniques have had on financial engineering.
	
	Given their universal approximation properties \cite{Cybenko, Hornik}, neural networks constitute a natural candidate for learning the implied-volatility map directly. Several studies have explored this direction. For example, \cite{LOB} proposed transformations of the option price aimed at improving approximation quality in difficult regions, while  \cite{DVPP} investigated a variety of neural architectures, including residual networks, highway networks and Deep Galerkin Method (\cite{SS}) networks. Their results reveal an important limitation: beyond a certain point, increasing network size does not necessarily improve accuracy and may even worsen generalisation performance. Moreover, the more sophisticated architectures do not overcome the same bottleneck precision encountered by very large standard feed-forward neural networks, their main advantage being only their much smaller training time.
	
	The present work is motivated by the observation that existing neural-network approaches largely treat implied-volatility inversion as a generic regression problem, while the implied-volatility function possesses a rich analytical structure that is already well understood through the work of Jäckel \cite{Jackel}, Tehranchi \cite{Tehranchi, Tehranchi2} and others. In particular, different regions of the input domain exhibit markedly different behaviours, especially in the asymptotic regimes corresponding to very small and very large option prices. We argue that a successful neural-network architecture should explicitly incorporate this information rather than rely solely on optimisation to discover it.
	
	The main contribution of this paper is therefore the introduction of a new family of \emph{asymptotically-informed neural-network architectures} for implied-volatility computation. Inspired by the domain partitioning analysis of \cite{Jackel}, we construct neural networks that learn a trainable partition of the price--log-moneyness domain and combine specialised local approximations of the implied-volatility function through a system of gating functions. The resulting architecture may be interpreted as a problem-specific mixture-of-experts model \cite{JJNH}, whose structure is directly derived from the analytical properties of the Black-Scholes formula.
	
	More specifically, the contributions of this paper are as follows:
	\begin{enumerate} 
	\item We design a new family of regime-aware neural architectures based on the asymptotic structure of the Black-Scholes pricing function. 
	\item We introduce trainable gating mechanisms that automatically learn the transition between asymptotic and central volatility regimes. 
	\item We demonstrate numerically that the proposed architectures consistently outperform standard feed-forward neural networks across a wide range of parameter domains, often by several orders of magnitude in relative accuracy. Thus, we provide empirical evidence that incorporating analytical knowledge into neural-network design yields substantially better accuracy and generalisation than treating implied-volatility inversion as a generic regression problem. 
	\item We show that the outputs of the proposed networks provide exceptionally accurate initial guesses for a third-order Householder scheme, allowing extremely high-precision implied-volatility computations after only a small number of iterations. 
	\item Unlike region-by-region rational approximations, the proposed approach scales naturally with model complexity and desired accuracy. Improving the approximation merely requires enlarging the training dataset or increasing the size of the neural network. Once trained, evaluations reduce essentially to matrix multiplications and pointwise non-linearities, operations that are highly optimised on modern hardware and naturally parallelisable. In contrast, increasing the accuracy of analytical approximation methods generally requires deriving new approximation formulas and evaluating increasingly complex expressions for each parameter regime separately.
	\end{enumerate}
	
	The remainder of the paper is organised as follows. In Section~\ref{sec:BSAnalysis}, we review several analytical properties of the Black-Scholes pricing function and derive the domain partition that motivates our architectures. Section~\ref{sec:Architectures} introduces the proposed neural-network constructions and discusses the different classes of local approximations and gating functions. Numerical experiments and comparisons with existing approaches are presented in Section~\ref{sec:NumericalExp}. Finally, Section~\ref{sec:Conclusion} concludes and discusses possible extensions of the methodology. Technical details concerning the numerically stable implementation of the Black-Scholes formula and the Householder refinement scheme are deferred to the appendices.
\section{Review of the analysis of the Black and Scholes formula}\label{sec:BSAnalysis}
	\subsection{Mathematical statement of the problem}
	In this subsection, we recall some properties of the Black-Scholes pricing formula that inspire the design of the neural network architectures developed in this work. Recall that the Black-Scholes formula for the price of a European call option with spot price $S_0$, strike $K$, time to maturity $T$, risk-free rate $r$ and volatility $\sigma$ is given by
	\[
	C(S_0,K,T,r,\sigma)=
	\Phi(d_1)S_0-\Phi(d_2)Ke^{-rT},
	\]
	where
	\[
	d_1=\frac{\log\left(\frac{S_0}{K}\right)+(r+\frac{\sigma^2}{2})T}{\sigma \sqrt{T}}
	\quad \text{and} \quad
	d_2=\frac{\log\left(\frac{S_0}{K}\right)+(r-\frac{\sigma^2}{2})T}{\sigma \sqrt{T}},
	\]
	and $\Phi$ denotes the standard Gaussian cumulative distribution function. This formula admits the following compact normalised representation, used for instance in \cite{Tehranchi, Tehranchi2}
	\begin{equation}\label{eq:BSpriceCompact}
	C_{\mathrm{BS}}(A,B)=
	\frac{C(S_0,K,T,r,\sigma)}{S_0}
	=\Phi\left(-\frac{A}{B}+\frac{B}{2}\right)
	-\Phi\left(-\frac{A}{B}-\frac{B}{2}\right)e^A,	
	\end{equation}
	where
	\[
	A= \log\left(\frac{Ke^{-rT}}{S_0} \right)
	\quad \text{and} \quad
	B= \sigma \sqrt{T}.
	\]
	$A$ is called the \emph{log-moneyness} and $B$ the \emph{total volatility} (or sometimes the integrated volatility). With these new notations and variables, and thinking of $B$ as a proxy for the volatility term $\sigma$, the implied volatility problem consists of determining the unique value $B^\ast>0$ of $B$ that, given an observed market price $C$ with a corresponding value of $A$ (computed using the option's parameters $S_0$, $K$,$T$ and $r$), solves the equation 
	\begin{equation}\label{eq:BinvDef}
	C_{\mathrm{BS}}(A,B)= \frac{C}{S_0}.
	\end{equation}
	Accordingly, we define the inversion map $B_{\mathrm{inv}}:(A,C)\longmapsto B^\ast$, which associates a log-moneyness and a normalised call price to the corresponding total volatility. 
	
	A useful symmetry relation follows directly from \eqref{eq:BSpriceCompact}
	\[ C_{\mathrm{BS}}(A,B) = e^A C_{\mathrm{BS}}(-A,B) + 1-e^A. \] 
	Consequently, it is sufficient to restrict our attention to the half-plane $A\geq 0$. In financial terms, this corresponds to considering out-of-the-money call options. 
	
	For ease of comparison with the seminal work \cite{Jackel}, we also recall the alternative normalisation used therein (but also in other works such as \cite{Li}):
	\[
	C_{\mathrm{J}}(\widetilde{A},B)=
	\frac{C(S_0,K,T,r,\sigma)}{\sqrt{S_0Ke^{-rT}}}
	=\Phi\left(\frac{\widetilde{A}}{B}+\frac{B}{2}\right)e^{\frac{\widetilde{A}}{2}}
	-\Phi\left(\frac{\widetilde{A}}{B}-\frac{B}{2}\right)e^{-\frac{\widetilde{A}}{2}},
	\]
	where $\widetilde{A}=-A= -\log\left(\frac{Ke^{-rT}}{S_0} \right)$. In particular, 
	$C_{\mathrm{J}}(\widetilde{A},B) = C_{\mathrm{BS}}(A,B) e^{-\frac{A}{2}}$.
	
	In order to simplify the formulae and without loss of generality, we shall henceforth assume that $S_0=1$ and $r=0$.

	From the perspective of neural network implementation, the normalisation $C_{\mathrm{BS}}$ is particularly advantageous. Indeed, for any fixed $A \geq 0$, the function $B \mapsto C_{\mathrm{BS}}(A,B)$ maps $(0,+\infty)$ into $(0,1)$ as
\[
\lim_{B\to 0^+} C_{\mathrm{BS}}(A,B) = 0
\qquad\text{and}\qquad
\lim_{B\to +\infty} C_{\mathrm{BS}}(A,B) = 1.
\]
	This means that normalised call prices are confined to the unit interval $(0,1)$ regardless of $A$. This bounded and parameter-independent range provides a particularly convenient input domain for neural-network training.

	In view of these observations, the implied-volatility problem reduces to approximating the function
\[
B_{\mathrm{inv}} : (A,C)\in [0,\infty) \times (0,1) \longmapsto B^\ast\in (0,\infty),
\]
	which maps a pair consisting of a log-moneyness $A$ and a normalised call price $C$ to the corresponding total volatility $B^\ast$.
	\subsection{Partitioning the price-log-moneyness space}
	In order to contextualise our approach and build intuition towards the proposed solution, we begin by reviewing the framework of \cite{Jackel}, where the aim is to construct an approximation of the inverse map $\widetilde B_{\mathrm{inv}} : (\widetilde A,\widetilde C) \to B$, defined implicitly through 
	\[ \widetilde C = C_{\mathrm J} \!\left( \widetilde A, \widetilde B_{\mathrm{inv}}(\widetilde A,\widetilde C) \right), \] 
	for given $\widetilde{A}\leq 0$ and \[ \widetilde{C} = \frac{C(S_0,K,T,r,\sigma)} {\sqrt{S_0Ke^{-rT}}} \in \left(0,e^{\widetilde A/2}\right). \]
	
	A key observation underlying most successful implied-volatility algorithms is that, for fixed log-moneyness, the Black-Scholes pricing function exhibits qualitatively different behaviours depending on the magnitude of the total volatility. For fixed $\widetilde A$, the function $B \mapsto C_{\mathrm J}(\widetilde A,B)$ is nearly flat for very small and very large values of $B$, while exhibiting an approximately linear behaviour in a neighbourhood of its inflection point. As a consequence, it is often advantageous to approximate the inverse map $\widetilde B_{\mathrm{inv}}$ separately in different regions of the $(\widetilde{A}, \widetilde{C})$-parameter space. More precisely, the $(\widetilde{A},B)$-plane (for $\widetilde{A}\leq 0$) is first partitioned into three regions of the form
	\begin{equation}\label{eq:3_regions_J}
	B\leq \widetilde{B}_{\lrm}(\widetilde{A})
	\;\;,\;\;
	\widetilde{B}_{\lrm}(\widetilde{A})\leq B\leq \widetilde{B}_{\urm}(\widetilde{A})
	\;\;\text{and}\;\;
	B \geq \widetilde{B}_{\urm}(\widetilde{A}),
	\end{equation}
	In \cite{Jackel}, the boundary functions $\widetilde B_{\mathrm{lrm}}$ and $\widetilde B_{\mathrm{urm}}$ are constructed using the tangent to the curve $B\mapsto C_{\mathrm J}(\widetilde A,B)$ at its inflection point $\widetilde{B}_{\crm}:=\sqrt{2|\widetilde{A}|}$with the asymptotic levels $\widetilde{C}=0$ and $\widetilde{C}=e^{\frac{\widetilde{A}}{2}}$ respectively. To avoid the repeated evaluation of the Black-Scholes function $C_{\mathrm{J}}$ (which can be computationally expensive and numerically unstable), \cite{GHMP} instead takes
$\widetilde{B}_{\lrm}$ and $\widetilde{B}_{\urm}$ to be affine functions of $\widetilde{A}$. This simplification yields a satisfactory approximation to the construction of \cite{Jackel} over a moderate range of log-moneyness values, but inevitably deteriorates for large $|\widetilde A|$ given the sublinear growth of the inflection point map $\widetilde{A} \mapsto \sqrt{2|\widetilde{A}|}$.

	Once the partition \eqref{eq:3_regions_J} has been constructed, different approximation techniques may be employed in each region. For example, \cite{Jackel} uses rational function approximations tailored to each regime (specifically, the rational cubic Hermite interpolants of \cite{DG}) whereas \cite{GHMP} constructs a single uniform approximation over the full $(\widetilde{A},\widetilde{C})$ domain using Chebyshev polynomials.	
	
	We now translate this framework in terms of the normalisation $C_{\mathrm{BS}}$ introduced in the previous subsection. Working in the domain $(A,C) \in [0,+\infty) \times (0,1)$, we apply the same geometric idea (partitioning via the tangent at the inflection point) to the function $B \mapsto C_{\mathrm{BS}}(A,B)$ for fixed $A \geq 0$. As in the $C_{\mathrm{J}}$ setting, this map exhibits an approximately linear behaviour near its inflection point $B_{\crm} := \sqrt{2A}$, and flat asymptotics at $B \to 0^+$ and $B \to +\infty$. This behaviour is illustrated in Figure~\ref{Fig:3_regions}-(A). The corresponding inverse function is shown in Figure~\ref{Fig:3_regions}-(B). Observe that the inverse becomes particularly steep near the two asymptotic regimes, reflecting the fact that a very small perturbation in the option price may result in a large variation in the implied volatility. These regions therefore represent the most challenging part of the inversion problem.
	
	\begin{figure}[h!]
  \centering
  \begin{subfigure}{0.48\textwidth}
    \centering
    \adjustbox{valign=c}{\includegraphics[width=\linewidth]{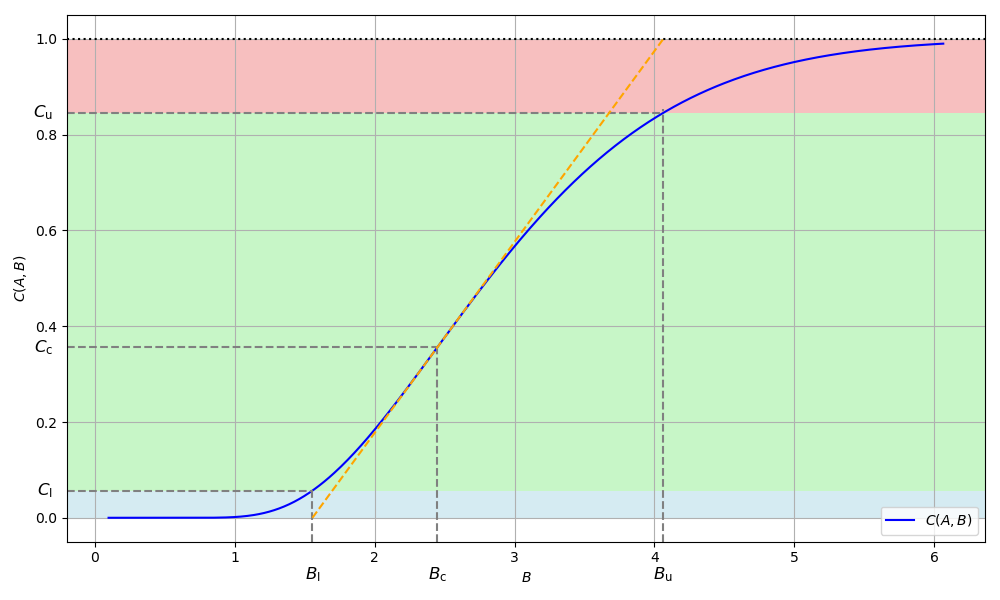}}
    \caption{ }
  \end{subfigure}
  \hfill
  \begin{subfigure}{0.48\textwidth}
    \centering
    \adjustbox{valign=c}{\includegraphics[width=\linewidth]{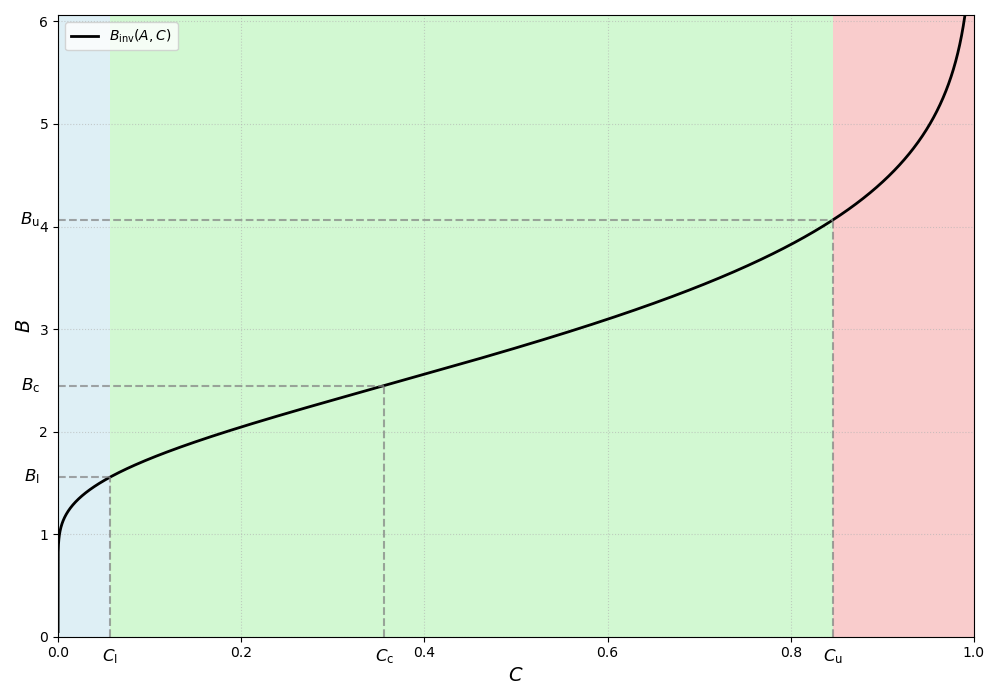}}
    \caption{ }
  \end{subfigure}
    \caption{(A): The partition of the space $(B,C)$ into three regions defined by the tangent to the curve $C_{\mathrm{BS}}(A,\cdot)$ at its inflection point, for $A=3$. (B): Inverse-volatility map $B_{\mathrm{inv}}(A,\dot)$ for $A=3$, and the corresponding partition of the $(C,B)$-domain into three regions defined by the tangent to the curve $C_{\mathrm{BS}}(A,\cdot)$ at its inflection point.}
    \label{Fig:3_regions}
	\end{figure}

%
	
	We now derive the boundary functions of the partition explicitly. For fixed $A \geq 0$, a direct computation gives the normalised price at the inflection point $B_{\crm} = \sqrt{2A}$ as
		\[
C_{\mathrm{BS}}\!\left(A,\sqrt{2A}\right)
= \tfrac{1}{2} - e^A\,\Phi\!\left(-\sqrt{2A}\right).
	\]
	The limiting behaviour as $A$ varies is readily obtained
	\[
	\lim_{A\to 0^+}C_{\mathrm{BS}}(A,\sqrt{2A})
	=0
	\quad \textrm{and}\quad
	\lim_{A\to + \infty}C_{\mathrm{BS}}(A,\sqrt{2A})
	=\frac{1}{2},
	\]
	where the second limit can be obtained for example using classical bounds for $\Phi$ of the type (see Theorem \ref{theo:MillsAsymptote})
	\begin{equation}\label{eq:GaussCumulIneq}
	\forall t>0\colon \quad
	\left( \frac{1}{t}-\frac{1}{t^3} \right) \varphi\left(t\right)
	\leq \Phi(-t) \leq 
	\frac{1}{t}  \varphi\left(t\right),	
	\end{equation}
	with $\varphi$ denotes the standard Gaussian density function.
	
	One further observes that the slope of $B \mapsto C_{\mathrm{BS}}(A,B)$ at its inflection point equals $1/\sqrt{2\pi}$, independently of $A$ (a particularly convenient feature of the $C_{\mathrm{BS}}$ normalisation). Denoting $B_{\urm}(A)$ the point of intersection of the asymptotic to $B\mapsto C_{\mathrm{BS}}(A,B)$ at its inflection point $\sqrt{2A}$ with the value $C_{\mathrm{BS}}=1$, then
	\[
		B_{\urm}(A) = \sqrt{2A} + \sqrt{\frac{\pi}{2}}
			+ \sqrt{2\pi} \Phi(-\sqrt{2A})e^A.
	\]
	In order to have a basic idea on the growth of the width of the asymptotic region $B\geq B_{\urm}(A)$, or, equivalently, $C\geq C_{\mathrm{BS}}(A,B_{\urm}(A))$, let us estimate (for small and large $A$) the value $C_{\mathrm{BS}}(A,B_{\urm}(A))$. One can easily show
	\[
	\lim_{A\to 0^+} C_{\mathrm{BS}}(A,B_{\urm}(A))
	= \Phi\left(\sqrt{\frac{\pi}{2}}\right)- \Phi\left(-\sqrt{\frac{\pi}{2}}\right)
	\simeq 0.7888,
	\]
	and that
	\[
	\lim_{A\to + \infty}C_{\mathrm{BS}}(A,B_{\urm}(A))= \Phi\left(\sqrt{\frac{\pi}{2}}\right)\simeq 0.895.
	\]
	Similarly, denoting $B_{\lrm}(A)$ the point of intersection of the tangent to $B\mapsto C_{\mathrm{BS}}(A,B)$ at its inflection point $\sqrt{2A}$ with the value $C_{\mathrm{BS}}=0$, then
	\[
	B_{\lrm}(A) = \sqrt{2A} - \sqrt{\frac{\pi}{2}}
			+ \sqrt{2\pi} \Phi(-\sqrt{2A})e^A.
	\]
	As before, we get
	\[
	\lim_{A\to 0^+} C_{\mathrm{BS}}(A,B_{\lrm}(A))
	=0,
	\]
	and
	\[
	\lim_{A\to +\infty} C_{\mathrm{BS}}(A,B_{\lrm}(A))
	= \Phi\left(-\sqrt{\frac{\pi}{2}}\right)\simeq 0.105.
	\]
	
	Altogether, the above analysis defines a natural partition of the $(A,C)$-domain into three regions (illustrated in Figure~\ref{Fig:3_regions_asymptotic}) which are characterised by the inequalities
	\begin{equation}\label{eq:3_regions}
	B\leq B_{\lrm}(A)
	\;\;,\;\;
	B_{\lrm}(A)\leq B\leq B_{\urm}(A)
	\;\;\text{and}\;\;
	B \geq B_{\urm}(A),
	\end{equation}

	\begin{figure}[h!]
	\includegraphics[scale=0.4]{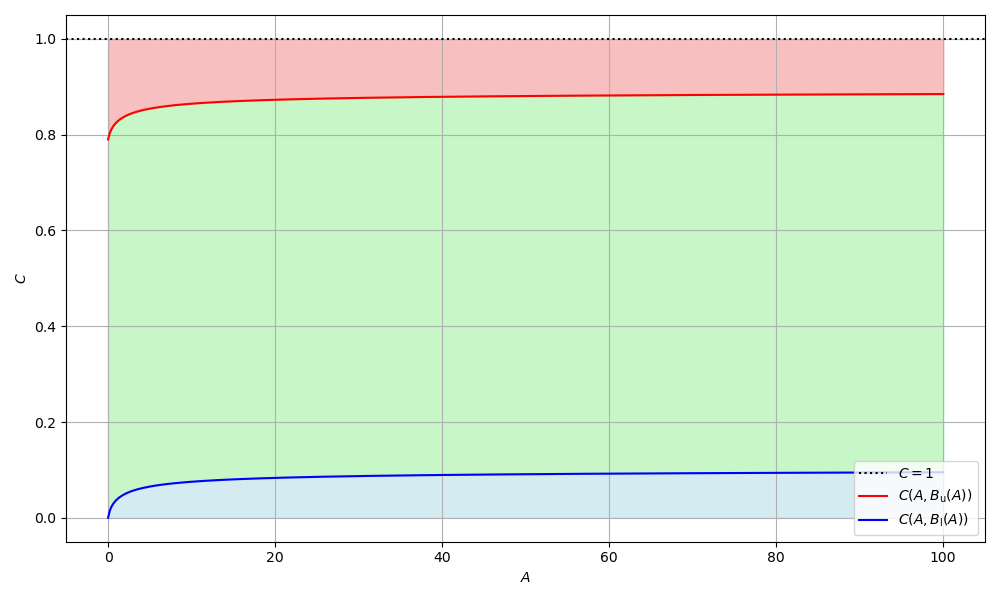}
	\centering
	\caption{Partition of the $(A,C)$-domain induced by the inequalities \eqref{eq:3_regions}. The central region corresponds to the approximately linear part of the pricing function, while the upper and lower regions correspond to the two asymptotic regimes.}
	\label{Fig:3_regions_asymptotic}
	\end{figure}
	
	This partition forms the basis of the neural-network architecture developed in this paper. Rather than learning a single global approximation of the total volatility map $B_{\mathrm{inv}}$, the model learns specialised approximations in each of the three regimes and combines them through a smooth transition mechanism. This allows the network to automatically adapt to the markedly different local behaviours of the implied-volatility function across the domain.
	\section{Architectural design}\label{sec:Architectures}
	In the remainder of this paper, we will use the term implied volatility to refer to the corresponding total volatility $B=\sigma\sqrt{T}$.
	\subsection{General principle}\label{subsec:ArchiGenPrincipal}
	The analysis of the previous section reveals that the implied-volatility function $B_{\mathrm{inv}}:[0,+\infty)\times(0,1)\longrightarrow(0,+\infty)$ exhibits markedly different behaviours in the three regions of the $(A,C)$-domain defined by the inequalities \eqref{eq:3_regions}. In particular, the inverse map becomes highly sensitive to the option price in the asymptotic regimes $C\to0^+$ and $C\to1^-$, whereas its behaviour is considerably more regular in the intermediate region. This observation suggests that learning the implied-volatility function with a single classical neural network may not be the most efficient strategy. Instead, we exploit the geometry of the problem and approximate $B_{\mathrm{inv}}$ by the following weighted combination
	\begin{equation}\label{eq:IVGenFormula}
	\mathrm{IV}(A,C)=f_0(A,C)g_0(A,C)+ f_1(A,C)g_1(A,C) + f_2(A,C)g_2(A,C),	
	\end{equation}	
	where the functions $(g_i)_{0\leq i \leq 2}$ are intended to approximate the implied volatility in different regions of the domain, while the functions $(f_i)_{0\leq i \leq 2}$ determine how these approximations are combined. Depending on the model under consideration, all of these functions may either be prescribed analytically or represented by trainable neural networks. We call $f_0$, $f_1$ and $f_2$ \emph{gating functions}. The role of $f_0$ and $f_1$ is to activate the approximations $g_0$ and $g_1$ in the lower and upper asymptotic regions, respectively. More precisely, we require $f_0$ and $f_1$ to be positive functions bounded by $1$ and satisfying for all $A\geq0$
	\begin{equation}\label{eq:ModFuncCond}
	f_0(A,C)\underset{C \to 0^+}{\longrightarrow} 1
	\;\;,\;\;
	f_0(A,C)\underset{C \to 1^-}{\longrightarrow} 0
	\;\;,\;\;
	f_1(A,C)\underset{C \to 0^+}{\longrightarrow} 0
	\;\;,\;\;
	f_1(A,C)\underset{C \to 1^-}{\longrightarrow} 1.
	\end{equation}
	
	The representation \eqref{eq:IVGenFormula} may be viewed as a problem-specific variant of a mixture-of-experts architecture (\cite{JJNH, Rokach, BCWR}). Unlike classical approaches to implied-volatility approximation, such as the region-wise constructions in \cite{Jackel,LOB}, the partition of the domain is not necessarily imposed a priori. Instead, it emerges during training through the learnt gating functions while still respecting the asymptotic constraints \eqref{eq:ModFuncCond}. In this sense, the architecture incorporates qualitative knowledge of the implied-volatility surface without hard-coding the boundaries of the different regimes.

	We shall investigate two variants of the decomposition \eqref{eq:IVGenFormula}, differing in the treatment of the third gating factor $f_2$.
	
	In the first variant, called interpolation model and labelled by ``\textbf{Inter}'', we impose $f_2:=1-f_0-f_1$. The resulting approximation takes the form
	\begin{equation}\label{eq:InterFormula}
	\textbf{Inter:}\quad 
	\mathrm{IV}(A,C)=f_0(A,C)g_0(A,C)+ f_1(A,C)g_1(A,C) + (1-f_0(A,C)-f_1(A,C))g_2(A,C).	
	\end{equation}
	In this setting, $(f_0,f_1,1-f_0-f_1)$ plays the role of a trainable parametric partition of unity on $\Rbb_+\times (0,1)$ (but it will not be one for all possible values of the parameters.) The approximation $g_2$ is therefore primarily responsible for modelling the implied volatility in the central region of the domain (corresponding roughly to the interval $[C_{\mathrm{BS}}(A,B_{\lbf}(A)),C_{\mathrm{BS}}(A,B_{\ubf}(A))]$ in \cite{Jackel}.)
	
	In the second variant, called free model and labelled by ``\textbf{Free}'', we simply set $f_2:=1$. The approximation then becomes
	\begin{equation}\label{eq:FreeFormula}
	\textbf{Free:}\quad 
	\mathrm{IV}(A,C)=f_0(A,C)g_0(A,C)+ f_1(A,C)g_1(A,C) + g_2(A,C).
	\end{equation}
	Under this formulation, $g_2$ acts as a global baseline approximation over the entire domain, whereas the terms $f_0g_0$ and $f_1g_1$ provide corrections tailored to the lower and upper asymptotic regions. 
	
	It is worth noting that decompositions of the form \eqref{eq:IVGenFormula} have conceptual similarities with several well-established neural-network architectures. For example, Highway Networks \cite{SGS} and Residual Networks \cite{HZRS} also learn representations based on combinations of multiple interacting components. However, in those architectures the decomposition is entirely determined by the optimisation procedure and typically admits no direct interpretation in terms of the underlying problem, thus excluding any a-priori knowledge. In contrast, the decomposition proposed here is motivated by the specific analytical structure of the implied-volatility function. Architectures related to Highway and Residual Networks were considered in the financial context for instance by \cite{DVPP} for implied-volatility approximation and model calibration. As will be seen in the numerical experiments, the incorporation of problem-specific structure allows our approach to achieve substantially higher accuracy while using significantly smaller networks.
	
	A different strategy for incorporating asymptotic information in the context of finance into neural-network approximations was proposed in \cite{AKP}. In this framework, a computationally expensive target function $f$ is decomposed as a sum of functions $S+\hat{f}$, where $S$ is a generalised spline matching a \emph{known} asymptotic of $f$ outside a prescribed learning region, and $\widehat f$ is a trainable linear combination of Gaussian kernels supported within that region. The methodology was applied to the approximation of Black-Scholes prices and SABR implied volatilities (with the assumption of uncorrelated Brownian motions driving the stock price and the volatility processes.) While this approach can be effective when accurate asymptotic formulae are available, it relies on substantial prior analytical knowledge of the target function. In many applications, asymptotic expansions are either unavailable or insufficiently accurate over practically relevant parameter ranges. For example, asymptotic formulae in option pricing are often valid only for extremely large strikes or maturities that lie far outside the region encountered in typical market data \cite{AL}. Moreover, the local nature of Gaussian-kernel corrections may require a large number of kernels to capture complex global features of the target function.
	
	The methodology developed in this paper takes a different perspective. Rather than relying on explicit asymptotic formulae, we exploit only qualitative information about the structure of the implied-volatility map, namely the existence of distinct asymptotic and central regimes. The resulting architecture remains fully data-driven while embedding the most important analytical features of the inversion problem directly into the model design.
	\subsection{Choices for the local implied volatility functions}\label{subsec:LocIVFct}
	Having described the overall architecture in the previous subsection, we now discuss the possible choices for the local approximation functions $g_i$, $i\in \{0,1,2\}$.
	
	The most direct approach is to model each function $g_i$ by a standard feed-forward neural network with scalar output. We refer to this class of models as ``\textbf{Gen}'' (for generic). In this setting, no structural constraint is imposed on the range of the network, and each $g_i$ is free to take arbitrary real values. The positivity of the final implied-volatility approximation is therefore enforced only indirectly through the training process.
	
	Since implied volatilities are strictly positive, it is natural to encode this property directly into the model. A simple way to achieve this is to define each local approximation $g_i$ as the exponential of the output of a feed-forward neural network. We refer to this family as ``\textbf{Exp}'' (for exponential). For the \textbf{Inter} architecture introduced in \eqref{eq:InterFormula}, this positivity constraint is naturally compatible with the intended interpretation of the functions $g_i$ as local implied-volatility approximations. The situation is slightly different for the \textbf{Free} architecture \eqref{eq:FreeFormula}. Recall that in this model the function $g_2$ acts as a global baseline approximation over the entire domain. Since the implied volatility in the lower asymptotic region ($C\to 0^-$) is smaller than in the central region, the correction term associated with this regime must decrease the value of the baseline approximation. Consequently, in the \textbf{FreeExp} model we impose for $g_0$ to be the negative of an exponentiated standard neural network output, which guarantees $g_0\leq 0$. The contribution $f_0g_0$ therefore acts as a negative correction in the lower asymptotic region, while the term $f_1g_1$ provides a positive correction in the upper asymptotic region.
	
	Inspired by \cite{Jackel}'s approach, a third non-obvious choice is to compute, based on asymptotic equivalents of the function $C_{\mathrm{BS}}(A,\cdot)$ at extreme values of $B$, proxies of the implied volatility function (corresponding to $g_0$ and $g_1$ in the setting of the \textbf{Inter} models) that are asymptotically linear in $A$ and $C$, which we can approximate using standard neural networks. Several variants of this idea were implemented and tested. The corresponding models, labelled \textbf{AsptDef} and \textbf{AsptAdv}, are available in the accompanying code repository. However, their empirical performance was consistently inferior to that of the other architectures considered in this paper. We believe that this behaviour can be explained by two main reasons. First, the asymptotic formulae are mathematically well-defined only on small subsets of the $(A,C)$-domain which depend on both variables simultaneously, and which do not naturally extend to globally defined functions on $ [0,+\infty)\times(0,1)$. As a consequence, additional transformations and regularisations are required before they can be incorporated into a neural-network architecture. Second, and as discussed in the previous subsection, many asymptotic implied-volatility formulae provide accurate approximations only in extremely remote regions of the parameter space and remain relatively inaccurate throughout the range of values relevant for training and testing. For these reasons, this third family of architectures will not be discussed further in the remainder of this paper.
	\subsection{Choices for the gating functions}\label{subsec:ModulIVFct}
	The success of the architectures proposed in Section \ref{subsec:ArchiGenPrincipal} depends not only on the quality of the local implied-volatility approximations $g_i$, but also on the ability of the gating functions to activate the appropriate local model in the appropriate region of the $(A,C)$-domain. It is therefore difficult to predict a priori which parametrisation of the gating functions will provide the best balance between flexibility, trainability and approximation accuracy, while also respecting the qualitative constraints laid out in \eqref{eq:ModFuncCond}. For this reason, we investigate several families of gating functions.
	
	Table \ref{tab:f0f1_models} summarises the families considered in this work. In all cases, $\alpha_i$, $\beta_i$, $\gamma_i$ and $\varepsilon_i$ are strictly positive trainable parameters (learnt as the exponentials of free trainable parameters), $\epsilon>0$ denotes a small fixed constant, and $N_f$ controls the expressive power of the gating functions. 
	\begin{table}[ht] 
	\centering 
	\footnotesize \renewcommand{\arraystretch}{2.2} \resizebox{\textwidth}{!}{
	\begin{tabular}{|c|c|c|} 
	\hline 
	\textbf{Model} & $f_0(A,C)$ & $f_1(A,C)$ \\ 
	\hline 
	\multirow{2}{*}{\textbf{PolyAC}} 
	& $\displaystyle f_{0,1,1}(A,C)= \frac{1} {1+ \sum\limits_{i=1}^{N_f} \alpha_i (A+\varepsilon_i)^{-\beta_i} \left(\frac{1}{C}-1\right)^{-\gamma_i}} $ 
	& \multirow{2}{*}{$\displaystyle f_{1,1}(A,C)= \frac{1} {1+ \sum\limits_{i=1}^{N_f} \alpha_i (A+\varepsilon_i)^{\beta_i} \left(\frac{1}{C}-1\right)^{\gamma_i}} $} \\
	 \cline{1-2} 
	 \textbf{PolyACeps} 
	 & $\displaystyle f_{0,1,2}(A,C)= \frac{1} {1+ \sum\limits_{i=1}^{N_f} \alpha_i (A+\epsilon)^{-\beta_i} \left(\frac{1}{C}-1\right)^{-\gamma_i}} $ 
	 & \\ 
	 \hline 
	 \textbf{PolyC} 
	 & $\displaystyle f_{0,2}(A,C)= \frac{1} {1+ \sum\limits_{i=1}^{N_f} \alpha_i \left(\frac{1}{C}-1\right)^{-\gamma_i}} $ 
	 & $\displaystyle f_{1,2}(A,C)= \frac{1} {1+ \sum\limits_{i=1}^{N_f} \alpha_i \left(\frac{1}{C}-1\right)^{\gamma_i}} $ \\ 
	 \hline 
	 \textbf{HardAC} 
	 & $\displaystyle f_{0,3}(A,C)= \frac{1} {1+ \sum\limits_{i=1}^{N_f} \alpha_i \max\!\left( \frac{C-C_{\mathrm{BS}}(A,B_{\lrm}(A))} {1-C}, 0 \right)^{\gamma_i}} $ 
	 & $\displaystyle f_{1,3}(A,C)= \frac{1} {1+ \sum\limits_{i=1}^{N_f} \alpha_i \max\!\left( \frac{C_{\mathrm{BS}}(A,B_{\urm}(A))} {C} -1, 0 \right)^{\gamma_i}} $ \\ 
	 \hline 
	 \multirow{2}{*}{\textbf{GaussAC}} 
	 & $\displaystyle f_{0,4,1}(A,C)= \exp\!\left( -\sum\limits_{i=1}^{N_f} \alpha_i (A+\varepsilon_i)^{-\beta_i} \left(\frac{1}{C}-1\right)^{-\gamma_i} \right) $ 
	 & \multirow{2}{*}{$\displaystyle f_{1,4}(A,C)= \exp\!\left( -\sum\limits_{i=1}^{N_f} \alpha_i (A+\varepsilon_i)^{\beta_i} \left(\frac{1}{C}-1\right)^{\gamma_i} \right) $} \\ 
	 \cline{1-2} 
	 \textbf{GaussACeps} 
	 & $\displaystyle f_{0,4,2}(A,C)= \exp\!\left( -\sum\limits_{i=1}^{N_f} \alpha_i (A+\epsilon)^{-\beta_i} \left(\frac{1}{C}-1\right)^{-\gamma_i} \right) $ 
	 & \\ 
	 \hline 
	 \textbf{GaussC} 
	 & $\displaystyle f_{0,5}(A,C)= \exp\!\left( -\sum\limits_{i=1}^{N_f} \alpha_i \left(\frac{1}{C}-1\right)^{-\gamma_i} \right) $ 
	 & $\displaystyle f_{1,5}(A,C)= \exp\!\left( -\sum\limits_{i=1}^{N_f} \alpha_i \left(\frac{1}{C}-1\right)^{\gamma_i} \right) $ \\ 
	 \hline 
	 \end{tabular} } 
	 \caption{Families of gating functions $f_0$ and $f_1$ considered in this work.} 
	 \label{tab:f0f1_models} 
	 \end{table}	
	 
	 All the parametrisations listed in Table~\ref{tab:f0f1_models} trivially satisfy the fundamental asymptotic requirements \eqref{eq:ModFuncCond}. In addition, they provide a smooth transition between the asymptotic and central regions of the domain. The principal distinction between the different families lies in the way the variables $A$ and $C$ are incorporated into the gating mechanism.
	 
	 The \textbf{PolyAC} and \textbf{GaussAC} families are the most flexible constructions. They allow the transition from one regime to another to depend simultaneously on the log-moneyness and the option price. This additional flexibility is motivated by the analysis of Section~\ref{sec:BSAnalysis}, where the location and width of the asymptotic regions depend strongly on $A$. 
	 
	 The \textbf{PolyC} and \textbf{GaussC} families depend exclusively on the price variable $C$. Their main advantage is their simplicity: they automatically satisfy the desired asymptotic behaviour while requiring fewer parameters than their AC counterparts. Furthermore, unlike the AC-based parametrisations, they do not suffer from potential distortions for very large values of $A$ ($f_{1,1}$ and $f_{1,4}$ converge to $0$ when $A \to +\infty$). Their main limitation is that the gating mechanism is identical for all levels of log-moneyness, despite the fact that the geometry of the implied-volatility surface varies substantially for small values of $A$ (see Figure~\ref{Fig:3_regions_asymptotic}).
	 
	 The \textbf{HardAC} family directly incorporates the asymptotic partition \eqref{eq:3_regions} introduced in Section~\ref{sec:BSAnalysis}. In particular, $f_{1,3}(A,C)=1$ whenever $C\ge C_{\mathrm{BS}}(A,B_{\urm}(A))$, and $f_{0,3}(A,C)=1$ whenever $C\le C_{\mathrm{BS}}(A,B_{\lrm}(A))$. Consequently, the corresponding local approximations become fully activated in the regions identified by the theoretical analysis inspired from \cite{Jackel}, while retaining trainable transitions between the different regimes. 
	 
	 Finally, the distinction between the polynomial and Gaussian families is primarily geometric. The polynomial models produce relatively slow transitions between regimes and may therefore capture broad overlap regions. In contrast, the Gaussian-type parametrisations decay exponentially and tend to generate sharper gating mechanisms. It is not clear a priori which behaviour is preferable for the implied-volatility inversion problem, and both possibilities are therefore investigated experimentally in the next section.
\section{Numerical experiments}\label{sec:NumericalExp}
	This section presents the experimental framework used to evaluate the proposed family of asymptotically-informed neural-network architectures for the computation of Black-Scholes implied-volatility. Our objectives are twofold. First, we assess the performance of the different variants introduced in Section~\ref{sec:Architectures}. Second, we compare these architectures against standard feed-forward neural networks and against a classical root-finding approach. 

	The experiments reported below should primarily be viewed as a study of the proposed architectural principles rather than as an exhaustive hyperparameter optimisation exercise. While key parameters such as the learning rate, the number of neurons per layer and the number of training epochs were explored, many additional design choices could be investigated further, including variations in network depth, activation functions, learning-rate schedules and layer widths. Moreover, the optimal architecture is expected to depend on the range of log-moneyness and volatility values represented in the training data.
	\subsection{Preparing the datasets}
	To learn the map $B_{\mathrm{inv}}$, we first build a dataset $\{(A_i,B_i,C_i)\}_{i=1}^m$, where $C=C_{\mathrm{BS}}(A,B)$, across a regular grid of values in the $(A,B)$-plane. As suggested by \cite{Jackel}, and instead of naive evaluations of the Black-Scholes formula, we use a numerically stable and highly accurate implementation of the price function $C_{\mathrm{BS}}$ that avoids catastrophic subtractive cancellation, underflow and overflow errors. The mathematical derivation and implementation details are presented in Appendix \ref{ap:AsymptoticBS}.
	
	In addition to the basic variables $(A,C)$, several auxiliary variables are pre-computed and supplied to selected architectures. These variables are intended to improve numerical stability and facilitate the learning of asymptotic behaviours:
	\[
	C_{\mathrm{inv}} = \frac{1}{C}-1,
	A_{\log}=\log(A),
	C_{\log}=\log(C_{\mathrm{inv}}),
	z_{\mathrm{u}} = \max\left( \frac{C_\urm}{C}-1, 0\right), 
	z_\lrm = \max \left( \frac{C- C_\lrm}{1-C} ,0\right).
	\]
	The quantities $C_{\mathrm{inv}}$, $A_{\log}$ and $C_{\log}$ naturally appear in several of the gating functions introduced in subsection \ref{subsec:ModulIVFct}. Similarly, the variables $z_{\urm}$ and $z_{\lrm}$ (which are only used in the \textbf{HardAC} architectures) encode the position of a sample relative to the asymptotic partition derived in Section \ref{sec:BSAnalysis}.
	
	$20\%$ of each dataset is reserved for testing. When adaptive learning-rate schedules are employed, an additional $15\%$ of the full dataset is reserved as a validation set which is used exclusively for learning-rate adjustment.
	\subsection{Neural networks with asymptotic regimes}
	We now describe the practical implementation of the architectures introduced in Section~\ref{sec:Architectures}. Each model is obtained by combining, through Equation \eqref{eq:IVGenFormula}, a choice of local implied-volatility functions $g_i$ (Section~\ref{subsec:LocIVFct}) with a choice of gating functions $f_i$ (Section~\ref{subsec:ModulIVFct}).
	
	Although the analytical expressions of the gating functions are simple, their implementation requires some care in order to ensure numerical stability and efficient training. To illustrate this point, consider the gating function $f_{1,1}$ of the PolyAC family with $N_f=1$ (which, despite the low value for $N_f$, produces better results than a standard neural network.) We implemented two mathematically equivalent but computationally distinct parametrisations, denoted by \textbf{Inv} and \textbf{Sig}.
	\begin{description}
	\item[Inv] We learn positive parameters $\alpha$, $\beta$, $\gamma$ and $\varepsilon$ and compute
	\begin{equation}\label{eq:PolyACInv}
	f_{1,1}(A,C)= \frac{1}{1+ \alpha (A+\varepsilon)^{\beta}C_{\mathrm{inv}}^{\gamma}}.
	\end{equation}
	\item[Sig] Instead of computing the reciprocal explicitly, we reformulate the expression by first computing
	\[
	t_{1,1} = \alpha - \beta A_{\log} - \gamma C_{\log},
	\]	
	where $\alpha \in \Rbb$, $\beta>0$ and $\gamma>0$ are trainable parameters. The gating function is then obtained via the application of the sigmoid activation $\sigma_{\mathrm{sig}}$ on $t_{1,1}$
	\[
	f_{1,1}(A,C)
	=\sigma_{\mathrm{sig}}(t_{1,1})
	= \frac{1}{1+ e^{-\alpha} A^{\beta}\left(\frac{1}{C}-1\right)^{\gamma}}.
	\]
	\end{description}
	The latter formulation is particularly appealing from an optimisation perspective since it exploits numerical primitives commonly used in deep-learning frameworks and avoids repeated evaluations of large powers and reciprocal transformations.
	
	The triple $G=(g_0,g_1,g_2)$ is produced by a feed-forward neural network with three hidden layers, each containing $N_g$ neurons and using the ReLU activation function. The inputs are the variables $(A,C)$ and the network outputs three scalar quantities corresponding to the functions $g_0$, $g_1$ and $g_2$.
	
	The final architectures are then obtained by combining the various modelling choices. For example, the architecture \textbf{PolyACInvExpInter} uses the \textbf{PolyAC} gating functions $f_{1,1}$ and $f_{0,1,1}$ implemented using the \textbf{Inv} method in \eqref{eq:PolyACInv}, while $G=(g_0,g_1,g_2)$ is learnt as the exponential of a standard neural network (to ensure positivity of the outputs). The final output of the network is then computed using the \textbf{Inter} formula \eqref{eq:InterFormula}. The architecture name therefore compactly encodes all design choices involved in the construction of the model.
	\subsection{Standard neural networks} To assess the benefit of the proposed asymptotic architectures, we compare them against conventional feed-forward neural networks (multi-layer perceptrons). These models are labelled \textbf{SimpleGen} and \textbf{SimpleExp}. Both networks take $(A,C)$ as inputs and consist of three hidden layers with ReLU activations. To ensure a fair comparison in terms of model complexity, each hidden layer contains $N=2N_g$ neurons. The two models differ only in their output layer. \textbf{SimpleGen} uses a standard affine output layer and therefore produces unrestricted real-valued outputs, while \textbf{SimpleExp} applies an exponential transformation to the final affine output, thereby enforcing positivity of the predicted implied volatility.
	\subsection{Learning and comparison metrics}\label{subsec:LossMetric} The quality of an implied-volatility approximation can be measured either in absolute terms or in relative terms. In particular, two loss functions are commonly used in quantitative-finance applications: the mean squared error (MSE) and the mean squared relative error (MSRE).
	\[\begin{array}{ccc}
	\mathrm{MSE}&=&\frac{1}{2k}\sum\limits_{i=1}^{k}|\mathrm{IV}(A_i,C_i)-B_i|^2,\\
	\mathrm{MSRE}&=&\frac{1}{2k}\sum\limits_{i=1}^{k}\left(\frac{\mathrm{IV}(A_i,C_i)-B_i}{B_i}\right)^2,\\
	\end{array}
	\]
	with $k$ denoting the size of the considered dataset (training, validation or testing.) 

	The choice between these losses has a substantial impact on the resulting approximation. Minimising the MSE tends to prioritise large values of the volatility, since errors in high-volatility regions dominate the optimisation process. As a consequence, excellent absolute accuracy may coexist with relatively poor performance in low-volatility regimes. Conversely, the MSRE assigns comparable importance to all relative errors and therefore typically yields superior accuracy when the volatility is small. Since low-volatility options are frequently encountered in practice, whereas option prices become increasingly insensitive to volatility when $C$ is close to its maximum value, the MSRE will be our default choice throughout this paper. Nevertheless, the provided code also handles the MSE error and the results of training all the proposed architectures with the MSE error are also provided in the code repository. To mitigate the errors in the large $B$ regime, one may also consider the following family of parametric losses which depend on the choice of two user-defined weights $\alpha$ and $\beta$
	\[
	\mathrm{MxLoss}(\alpha,\beta)=\alpha\,\mathrm{MSE} + \beta \,\mathrm{MSRE}.
	\]
	Although this loss function is not discussed further here, it is available in the accompanying implementation.
	
	Regardless of the chosen loss function during training, all models are evaluated using the following performance metrics:
	\begin{enumerate}
	\item Mean Squared Error $\mathrm{MSE}$.
	\item Mean Squared Relative Error $\mathrm{MSRE}$.
	\item Maximum absolute residual error
	\[
	\Delta=\max_{1\leq i \leq k} |\mathrm{IV}(A_i,C_i)-B_i|.
	\]
	\item Maximum relative residual error
	\[
	\delta=\max_{1\leq i \leq k} \frac{|\mathrm{IV}(A_i,C_i)-B_i|}{B_i} .
	\] 
	\end{enumerate}
	
	To assess robustness and training variability, each experiment was repeated independently $n=40$ times. The comparison of the scores above between training and testing datasets serve to ensure that the functions learnt generalise well to unseen data. 
	
	The neural networks were trained using three variations of the Adam optimiser\footnote{\href{https://docs.pytorch.org/docs/2.12/generated/torch.optim.Adam.html}{https://docs.pytorch.org/docs/2.12/generated/torch.optim.Adam.html}} with an initial learning rate $\eta=10^{-3}$ and $150$ epochs:
	\begin{itemize}
		\item Full-batch Adam,
		\item Mini-batch Adam with batch size $B=256$ and constant learning rate,
		\item Mini-batch Adam with batch size $B=128$ and adaptive learning-rate reduction based on validation performance. The reduction factor was $0.25$, the tolerance threshold $10^{-2}$ and the patience window $5$ epochs.
	\end{itemize}
	The third strategy consistently produced the best results and is therefore the only one reported in the remainder of the paper. Results obtained using the alternative training schemes are available in the code repository.
	\subsection{The Householder root finding algorithm} 
	As will be demonstrated in the following subsection, the proposed architectures provide highly accurate approximations of the implied volatility. Nevertheless, these approximations can be refined further at negligible computational cost by using them as initial guesses for a root-finding algorithm. This approach, which we will follow without any additional innovation, is standard in the implied-volatility literature and combines the speed of a neural-network approximation with the accuracy of an iterative numerical solver. 
	
	Following \cite{Jackel}, we employ the third-order Householder method \cite{Householder}. Based on an initial guess for $B$, the algorithm iteratively solves $C_{\mathrm{BS}}(A,B)=C_0$, where $C_0$ denotes the observed option price. For extremely small or extremely large values of the volatility, directly applying Householder's method to the objective function $g:=C_{\mathrm{BS}}(A,\cdot)-C_0$ may lead to numerical instabilities. To overcome this difficulty, we opt for the alternative objective functions proposed in \cite{Jackel}. The detailed derivation of the resulting algorithm is given in Appendix \ref{ap:Householder}.
	\subsection{Numerical results}
	We now evaluate the practical performance of the architectures introduced in Section~\ref{sec:Architectures}. The objective of this subsection is twofold. First, we assess the extent to which the proposed asymptotic decomposition improves the direct approximation of the implied volatility. Second, we investigate how these approximations behave when used as initial guesses for a high-order root-finding algorithm. 
	
	To facilitate comparison with the classical implied volatility literature, and in particular with the benchmark studies of \cite{Jackel}, we consider ranges of log-moneyness and total volatility that encompass both practically relevant and highly challenging asymptotic regimes.
	
	For each dataset, all neural network architectures proposed in this paper were trained and evaluated over $40$ independent runs. The complete collection of results is available in the accompanying code repository
	\footnote{The full numerical results and the code to replicate all the figures and tables in this paper are publicly available on \href{https://github.com/samira-amiriyan/asymptotically-informed-implied-volatility}{https://github.com/samira-amiriyan/asymptotically-informed-implied-volatility}.}. To preserve readability, we only report here the strongest performing architectures together with the best standard feed-forward neural network. The retained and discussed architectures are the ones that scored consistently low values across all metrics introduced in Section~\ref{subsec:LossMetric}; in particular, these show good generalisation properties, as measured by the proximity between training and testing errors. Whenever possible, and as long as the performance was comparable, we chose a representative of each one of the gating function families presented in Subsection~\ref{subsec:ModulIVFct}.
	
	The experiments reported below correspond to a second series of $25$ independent training runs of the selected best performing architectures, from which the reported statistics were computed. A third series of $10$ runs was then used to identify and store the best-performing set of network parameters for each architecture, which were subsequently used to generate the visualisations.
	
	Below, all architectures use 
	\begin{itemize}
	\item $N_f=5$, the number of terms used in the gating functions (unless the implementation of the gating function is of the \textbf{Sig} type, which does not use this parameter).
	\item $N_g=64$, the hidden width of the local implied-volatility networks.
	\item $N=128$, the width of the standard feed-forward networks.
	\end{itemize}
	\subsubsection{A large-range dataset}	
	We begin with a particularly demanding benchmark designed to test the robustness of the proposed architectures well beyond the range typically encountered in market applications. More precisely, we consider $(A,B)\in [0,16]\times [10^{-5},7.07]$. This range covers several extreme asymptotic regimes and therefore provides a stringent test of any implied-volatility approximation method.
	
	Figure~\ref{fig:evolution_loss_sgd_var_DS1large} displays the evolution of the training loss averaged over $25$ runs. Table~\ref{tab:loss_sgd_var_DS1large} reports the corresponding numerical performance metrics on both the training and testing datasets before any root-finding refinement is applied.
	
	Several fundamental observations immediately emerge from Figure \ref{fig:evolution_loss_sgd_var_DS1large} and Table~\ref{tab:loss_sgd_var_DS1large}. First, all asymptotic architectures dramatically outperform the standard feed-forward neural network. The improvement is particularly striking for the MSRE metric, where the best asymptotic architectures reduce the error by roughly four orders of magnitude. Similar gains are observed in the maximum relative error metric. This confirms that explicitly modelling the asymptotic structure of the inversion problem provides a substantial advantage over learning the inverse map directly with a conventional neural network. Second, all selected asymptotic architectures exhibit excellent generalisation properties. For every reported metric, training and testing performances are nearly identical, indicating that the observed gains do not arise from overfitting but from a genuinely better approximation of the underlying inverse map. Third, despite its apparent simplicity, the architecture PolyCInvExpInter remains highly competitive. This observation is consistent with the geometric analysis of Section~\ref{sec:BSAnalysis}. Over a substantial portion of the domain, the asymptotic boundaries become nearly horizontal in the $(A,C)$-plane, making a gating mechanism depending only on $C$ surprisingly effective. Nevertheless, architectures incorporating both variables generally achieve the best overall balance across metrics. Finally, it is worth emphasising that, even before any iterative refinement, the proposed architectures achieve accuracy levels substantially beyond those reported by classical closed-form approximations such as~\cite{Li}, while remaining valid over a significantly wider range of parameters.
	\begin{figure}[h!]
	\includegraphics[scale=0.4]{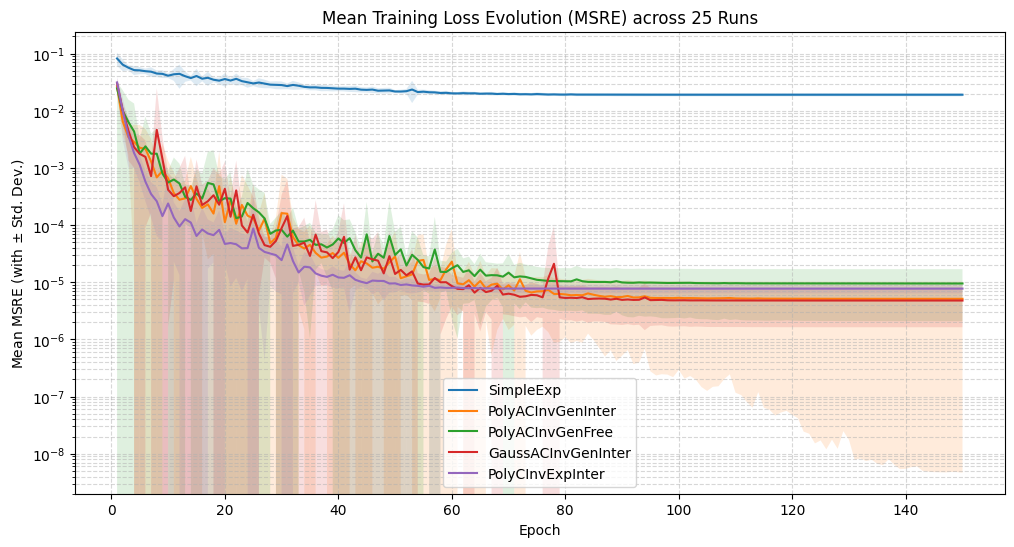}
	\centering
	\caption{Evolution of the training MSRE averaged over $25$ independent runs for the selected architectures on the dataset $(A,B)\in[0,16]\times[10^{-5},7.07]$. The shaded regions represent one standard deviation around the mean.}
	\label{fig:evolution_loss_sgd_var_DS1large}
	\end{figure}
	
\begin{table}[ht]
	\centering
	\scriptsize
	\setlength{\tabcolsep}{3pt}
	\renewcommand{\arraystretch}{1.15}
	
	\begin{adjustbox}{max width=\textwidth}
		\begin{tabular}{lcccccccc}
			\toprule
			Model 
			& \makecell{$\mathrm{MSE}_\mathrm{train}$} 
			& \makecell{$\mathrm{MSE}_\mathrm{test}$} 
			& \makecell{$\mathrm{MSRE}_\mathrm{train}$} 
			& \makecell{$\mathrm{MSRE}_\mathrm{test}$} 
			& \makecell{$\Delta_\mathrm{train}$} 
			& \makecell{$\Delta_\mathrm{test}$} 
			& \makecell{$\delta_\mathrm{train}$} 
			& \makecell{$\delta_\mathrm{test}$} \\
			\midrule
			
			SimpleExp 
			& $4.55 \times 10^{-2}$
			& $4.57 \times 10^{-2}$ 
			& $1.91 \times 10^{-2}$ 
			& $1.88 \times 10^{-2}$ 
			& $1.11$ 
			& $1.11$
			& $9.85 \times 10^{-1}$
			& $1.72$ \\
			
			PolyACInvGenInter  
			& $\mathbf{3.08 \times 10^{-5}}$ 
			& $\mathbf{3.14 \times 10^{-5}}$ 
			& $5.13 \times 10^{-6}$ 
			& $5.29 \times 10^{-6}$
			& $6.96 \times 10^{-2}$
			& $6.84 \times 10^{-2}$
			& $4.98 \times 10^{-2}$ 
			& $7.46 \times 10^{-2}$ \\
			
			PolyACInvGenFree  
			& $4.51 \times 10^{-5}$ 
			& $4.59 \times 10^{-5}$ 
			& $9.52 \times 10^{-6}$ 
			& $9.86 \times 10^{-6}$ 
			& $6.44 \times 10^{-2}$ 
			& $6.33 \times 10^{-2}$ 
			& $5.86 \times 10^{-2}$ 
			& $1.06 \times 10^{-1}$ \\
			
			GaussACInvGenInter  
			& $3.21 \times 10^{-5}$ 
			& $3.27 \times 10^{-5}$ 
			& $\mathbf{4.79 \times 10^{-6}}$
			& $\mathbf{4.89 \times 10^{-6}}$
			& $7.65 \times 10^{-2}$ 
			& $7.51 \times 10^{-2}$ 
			& $\mathbf{4.49 \times 10^{-2}}$ 
			& $\mathbf{6.49 \times 10^{-2}}$ \\
			
			PolyCInvExpInter   
			& $3.58 \times 10^{-5}$
			& $3.59 \times 10^{-5}$ 
			& $7.74 \times 10^{-6}$ 
			& $7.71 \times 10^{-6}$ 
			& $\mathbf{3.56 \times 10^{-2}}$ 
			& $\mathbf{3.50 \times 10^{-2}}$ 
			& $1.30 \times 10^{-1}$ 
			& $6.87 \times 10^{-2}$ \\
			
			\bottomrule
		\end{tabular}
	\end{adjustbox}
	
	\caption{Average performance of the selected neural networks across $25$ independent runs on the dataset $(A,B)\in[0,16]\times[10^{-5},7.07]$ (trained with the MSRE loss function minimised using mini-batch Adam with adaptive learning-rate.)}
	\label{tab:loss_sgd_var_DS1large}
\end{table}

	While aggregate metrics provide a useful global summary, they do not reveal where errors occur in the parameter space. To address this point, Figure~\ref{fig:heatmap_AB_MSRE_DS1large_0_16_1000_1em5_7p07_1000} displays heat maps of the pointwise relative error (on a logarithmic scale) over a dense evaluation grid, containing $1000$ points in each coordinate direction, which is twice the resolution used for the initial dataset. Grey regions correspond to parameter values producing numerically meaningless prices outside the interval $[10^{-50}, 1-10^{-50}]$. 
	
	The contrast between the standard feed-forward network and the proposed architectures is particularly revealing. The standard neural network exhibits a large region of elevated error for small values of the volatility, the range of which increases with log-moneyness. This is precisely the regime where implied-volatility inversion is most challenging and where the asymptotic analysis of Section~\ref{sec:BSAnalysis} suggests specialised treatment is necessary. In contrast, the error produced by the asymptotic architectures remains uniformly small across almost the entire admissible domain. The heat maps therefore demonstrate that the improvement is not confined to isolated regions but is genuinely global.

	\begin{figure}[h!]
	\includegraphics[width=\linewidth]{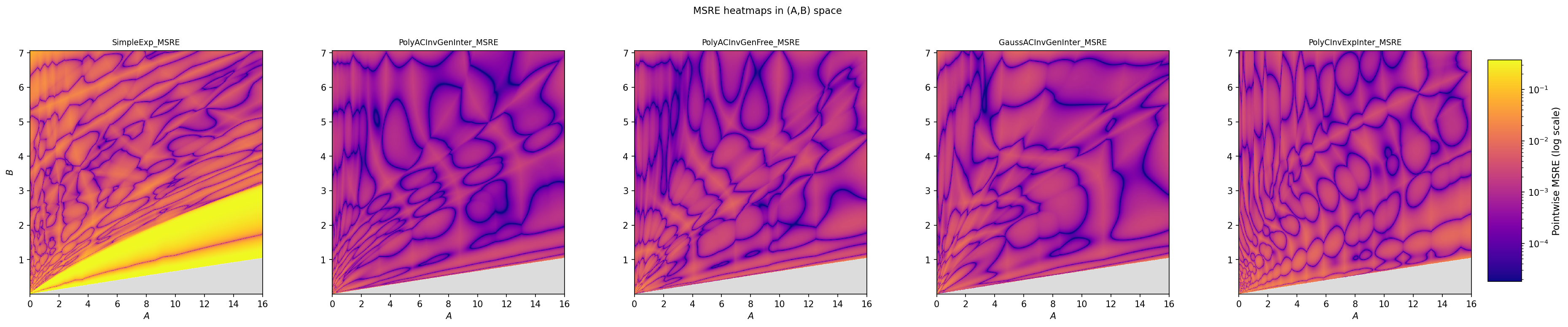}
	\centering
	\caption{Pointwise relative error, on a logarithmic scale, of the chosen architectures over the domain $(A,B)\in[0,16]\times[10^{-5},7.07]$. The evaluation grid contains $1000\times1000$ points.
	}
	\label{fig:heatmap_AB_MSRE_DS1large_0_16_1000_1em5_7p07_1000}
	\end{figure}
	
	The previous results assessed the neural networks as stand-alone implied-volatility approximators. In practice, however, one can exploit them even more effectively by using their outputs as initial guesses for an iterative root-finding procedure. Given $(A,C)$, denote by $B_0$ the output of a neural network for the implied volatility, then by $B_1$ and $B_2$ the refinement of this value by the Householder method (described in Appendix~\ref{ap:Householder}) after $1$ and $2$ iterations respectively. To quantify convergence, we consider the logarithm of the relative error
	\begin{equation}\label{eq:LogRelError}
	\log \max \left( \left|\frac{B_i}{B_{\mathrm{inv}}} -1\right| , 10^{-18}\right),
	\end{equation}
	
	Figure~\ref{fig:log_convergence_Householder_2_iter_DS1large_A_5} shows the convergence behaviour for a fixed log-moneyness value $A=5$ over the interval $B\in [10^{-5},2]$. The figure demonstrates that the neural-network outputs already provide highly accurate initial approximations. After only two Householder iterations, all proposed asymptotic architectures reach essentially the maximum precision across the entire range investigated. In contrast, the standard neural network frequently starts from a significantly poorer initial approximation, particularly for small volatilities. While Householder iterations still improve the prediction substantially, the resulting accuracy remains noticeably inferior.
	
	\begin{figure}[h!]
	\includegraphics[width=\linewidth]{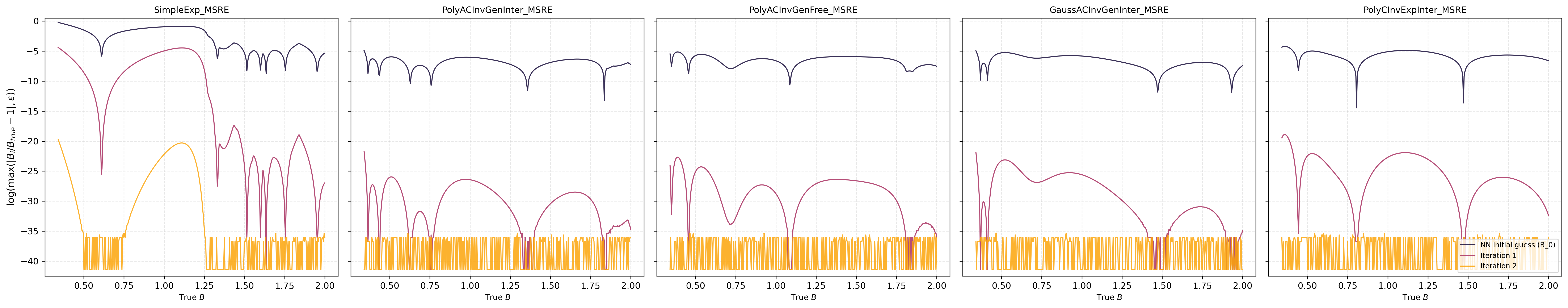}
	\centering
	\caption{Evolution of the logarithmic relative error \eqref{eq:LogRelError} for a fixed value $A=5$ over the interval $B\in [10^{-5},2.00]$. The curves correspond to the neural-network output ($B_0$), after one Householder iteration, and after two Householder iterations. $B$ is sampled uniformly over a grid of $500$ points.}
	\label{fig:log_convergence_Householder_2_iter_DS1large_A_5}
	\end{figure}
	
	The corresponding numerical evidence is quantified in Table~\ref{tab:stats_Householder_2_iter_DS1large_A_5}. The average logarithmic relative error after two iterations is close to the numerical precision limit for all asymptotic architectures, while the maximum observed error remains extremely small throughout the test range. The similarity between the different asymptotic architectures is itself noteworthy. Once a sufficiently accurate initial approximation has been obtained, only a small number of Householder iterations are required to eliminate the remaining discrepancy.
	
	\begin{table}[ht]
  \centering
  \begin{tabular}{lccc}
    \hline
    Model & Avg Log Error & Std Log Error & Max Log Error \\
    \hline
    SimpleExp                & $-33.56$ & $7.38$ & $-19.70$ \\
    PolyACInvGenInter        & $\mathbf{-38.85}$ & $2.61$ & $\mathbf{-35.35}$ \\
    PolyACInvGenFree         & $\mathbf{-38.85}$ & $2.60$ & $\mathbf{-35.35}$ \\
    GaussACInvGenInter       & $-38.72$ & $\mathbf{2.59}$ & $\mathbf{-35.35}$ \\
    PolyCInvExpInter         & $-38.70$ & $2.62$ & $\mathbf{-35.35}$ \\
    \hline
  \end{tabular}
    \caption{Average, standard deviation and the maximum of the logarithm of the relative error \eqref{eq:LogRelError} for the chosen architectures after $2$ Householder iterations with $A=5$ and $B\in [10^{-5},2.00]$. $B$ is sampled uniformly over a grid of $500$ points.}
    \label{tab:stats_Householder_2_iter_DS1large_A_5}
\end{table}

	To assess the robustness of this behaviour throughout the entire parameter space, Figure~\ref{fig:heatmap_log_rel_err_nobrent_DS1large_0_16_1em5_7p07_1000} displays the logarithm of the relative error \eqref{eq:LogRelError} after two Householder iterations over the full domain $(A,B)\in[0,16]\times[10^{-5},7.07]$, with $1000$ grid values for each variable. The difference between the standard neural network and the proposed architectures remains pronounced in the low-volatility region.

	\begin{figure}[h!]
	\includegraphics[width=\linewidth]{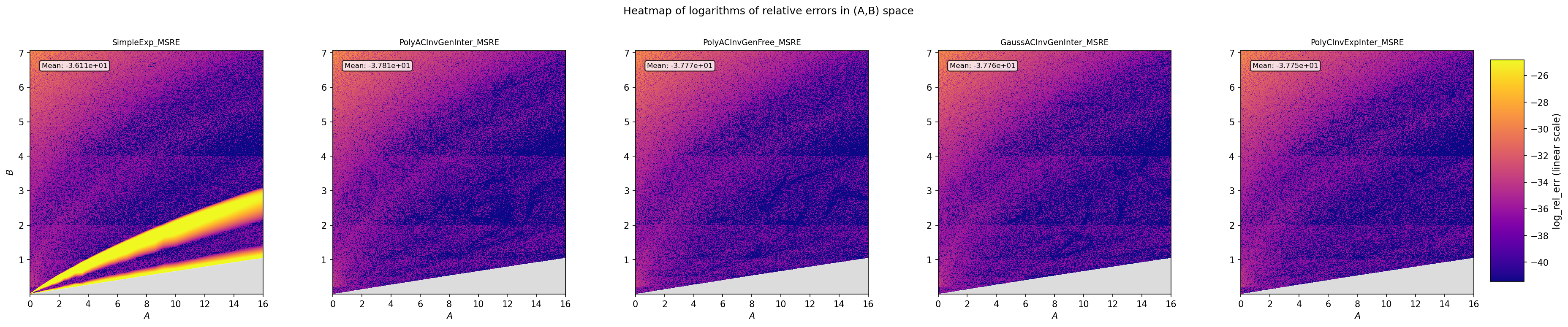}
	\centering
	\caption{Heat map of the logarithm of the relative error \eqref{eq:LogRelError} of the chosen neural networks after $2$ iterations of the Householder root finding method over the range of values $(A,B)\in[0,16]\times[10^{-5},7.07]$, with $1000$ grid values for each variable.}
	\label{fig:heatmap_log_rel_err_nobrent_DS1large_0_16_1em5_7p07_1000}
	\end{figure}
	
	As a final benchmark, we compare the neural-network-plus-Householder methodology with a classical root-finding strategy based on Brent's method\footnote{as implemented in the scipy.optimize library: \\
	\href{https://docs.scipy.org/doc/scipy/reference/generated/scipy.optimize.brentq.html}{https://docs.scipy.org/doc/scipy/reference/generated/scipy.optimize.brentq.html}} \cite{Brent}. Because of the computational cost associated with dense evaluations of Brent's algorithm, we restrict our attention to the more realistic region $(A,B)\in[0,0.5]\times[10^{-5},2]$, sampled on a $200\times200$ grid. Brent's method is allowed a maximum of $50$ iterations.
	
	Figure~\ref{fig:heatmap_log_rel_err_brent_DS1large_0_0p5_1em5_2_200} compares the logarithmic relative error \eqref{eq:LogRelError} obtained by Brent's method with that obtained by the proposed neural-network-plus-Householder pipeline. While Brent's algorithm is generally regarded as highly reliable for moderate values of the implied volatility, the resulting error landscape appears considerably less structured. Regions of reduced accuracy occur in a manner that is difficult to predict from the geometry of the parameter space. In contrast, the asymptotic architectures provide highly accurate initial guesses and therefore benefit from the rapid local convergence of Householder's method.

	\begin{figure}[h!]
	\includegraphics[width=\linewidth]{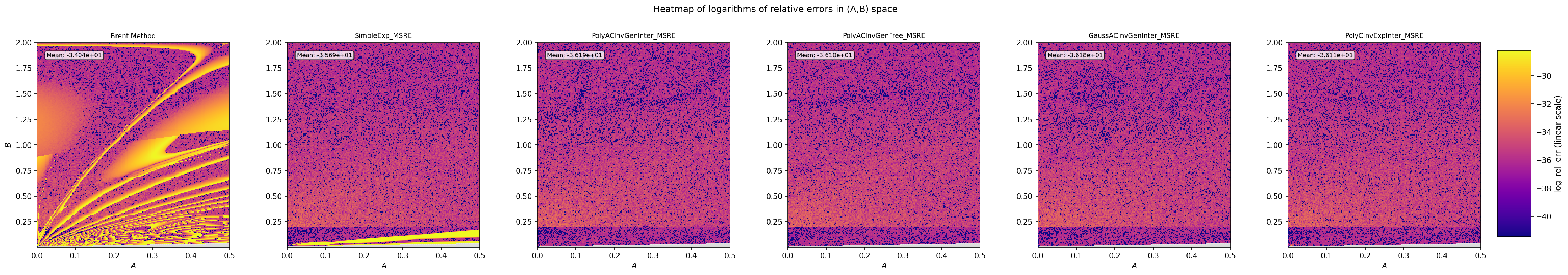}
	\centering
	\caption{Comparison of the logarithm of the relative error \eqref{eq:LogRelError} between Brent's method (maximum of $50$ iterations) and the proposed neural-network-plus-Householder approach on the domain $(A,B)\in[0,0.5]\times[10^{-5},2]$. The evaluation grid contains $200\times 200$ points.}
	\label{fig:heatmap_log_rel_err_brent_DS1large_0_0p5_1em5_2_200}
	\end{figure}
	
	These results suggest that, rather than completely replacing classical root-finding methods, the proposed neural networks can act as powerful accelerators that place the iterative solver immediately in the convergence regime. The combination of an asymptotically-informed neural network with only two Householder iterations therefore yields an extremely accurate and computationally efficient implied-volatility solver.
	\subsubsection{A medium-range dataset}	
	We now consider a more realistic range of parameters, $(A,B)\in[0,3]\times[10^{-7},1.22]$, which contains the region most frequently encountered in practical implied-volatility computations while still including challenging low-volatility regimes. The corresponding training curves, summary statistics, error heat maps and Householder refinements are reported in Figures~\ref{fig:evolution_loss_sgd_var_DS2medium}, \ref{fig:heatmap_AB_MSRE_DS2medium_0_3_1000_1em7_1p22_1000}, \ref{fig:log_convergence_Householder_2_iter_DS2medium_A_1p5} and \ref{fig:heatmap_log_rel_err_brent_DS2medium_0_3_1em7_1p22_200}, together with Tables~\ref{tab:loss_sgd_var_DS2medium} and \ref{tab:stats_Householder_2_iter_DS2medium_A_1p5}.
	
	The conclusions largely reinforce those obtained on the larger benchmark dataset. All asymptotic architectures substantially outperform the standard feed-forward neural network across every metric considered. For example, the best-performing architecture, \textbf{GaussACInvExpInter}, reduces the test MSRE from $3.51\times10^{-2}$ to $ 2.24\times10^{-5}$, corresponding to an improvement by more than three orders of magnitude. Similar gains are observed for the MSE and maximum residual metrics. Furthermore, the close agreement between training and testing errors once again confirms the excellent generalisation properties of the proposed architectures.
	
	A notable difference compared with the previous dataset is the emergence of a clear preference for gating functions that depend jointly on the log-moneyness $A$ and the option price $C$. Indeed, all of the strongest-performing architectures belong to the AC family. This observation is consistent with the geometric analysis of Section~\ref{sec:BSAnalysis}. As $A\to 0$, the location of the asymptotic boundaries varies significantly with $A$, making a purely price-based partition of the domain less effective than in the large-range benchmark. 
	
	Another noteworthy observation is the strong performance of the Gaussian family of gating functions (in particular the model GaussACInvExpInter), which dominate nearly all reported metrics and exhibit remarkably stable behaviour across repeated runs. The polynomially-induced gating functions remain competitive, however. In particular, the architecture PolyACSigExpFree appears among the strongest-performing models despite relying on the \textbf{Sig} implementation, which effectively corresponds to the case $N_f=1$. Although its accuracy remains slightly below that of the most expressive Gaussian models, its performance demonstrates that even relatively simple asymptotic parametrisations can provide substantial improvements over standard neural networks.
	
	We also remark that the achieved relative accuracies remain comparable to those obtained on the substantially larger dataset of the previous subsection. Since the current dataset contains much smaller volatility values, one might expect relative errors to increase significantly. The fact that this does not occur suggests that the proposed architectures genuinely learn the asymptotic structure of the inversion map rather than merely interpolating over the training domain.

	\begin{figure}[h!]
	\includegraphics[scale=0.4]{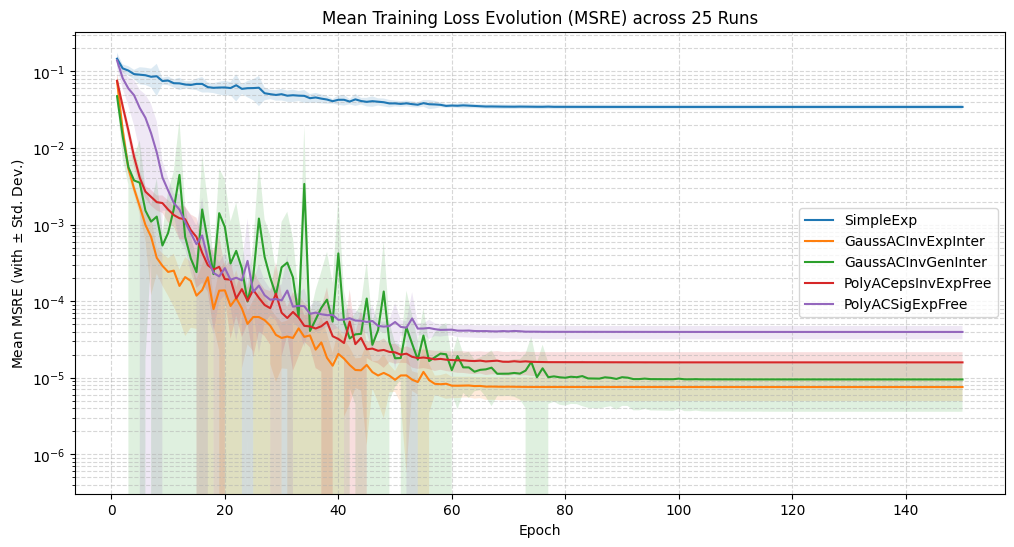}
	\centering
	\caption{Evolution of the training MSRE averaged over $25$ independent runs for the selected architectures on the dataset $(A,B)\in [0,3]\times[10^{-7},1.22]$. The shaded regions represent one standard deviation around the mean.}
	\label{fig:evolution_loss_sgd_var_DS2medium}
	\end{figure}	
	
\begin{table}[ht]
	\centering
	\scriptsize
	\setlength{\tabcolsep}{3pt}
	\renewcommand{\arraystretch}{1.15}
	
	\begin{adjustbox}{max width=\textwidth}
		\begin{tabular}{lcccccccc}
			\toprule
			Model 
			& \makecell{$\mathrm{MSE}_\mathrm{train}$} 
			& \makecell{$\mathrm{MSE}_\mathrm{test}$} 
			& \makecell{$\mathrm{MSRE}_\mathrm{train}$} 
			& \makecell{$\mathrm{MSRE}_\mathrm{test}$} 
			& \makecell{$\Delta_\mathrm{train}$} 
			& \makecell{$\Delta_\mathrm{test}$} 
			& \makecell{$\delta_\mathrm{train}$} 
			& \makecell{$\delta_\mathrm{test}$} \\
			\midrule
			
			SimpleExp 
			& $4.49 \times 10^{-3}$
			& $4.52 \times 10^{-3}$
			& $3.44 \times 10^{-2}$ 
			& $3.51 \times 10^{-2}$
			& $2.97 \times 10^{-1}$ 
			& $2.96 \times 10^{-1}$ 
			& $1.11$
			& $1.10$ \\
			
			GaussACInvExpInter  
			& $\mathbf{1.30 \times 10^{-6}}$
			& $\mathbf{1.30 \times 10^{-6}}$
			& $\mathbf{7.53 \times 10^{-6}}$
			& $2.24 \times 10^{-5}$ 
			& $\mathbf{6.92 \times 10^{-3}}$
			& $\mathbf{6.64 \times 10^{-3}}$ 
			& $1.67 \times 10^{-1}$ 
			& $8.21 \times 10^{-1}$ \\
			
			GaussACInvGenInter  
			& $1.69 \times 10^{-6}$
			& $1.70 \times 10^{-6}$ 
			& $9.46 \times 10^{-6}$ 
			& $\mathbf{1.04 \times 10^{-5}}$ 
			& $7.26 \times 10^{-3}$ 
			& $7.09 \times 10^{-3}$ 
			& $\mathbf{1.19 \times 10^{-1}}$ 
			& $\mathbf{1.53 \times 10^{-1}}$ \\
			
			PolyACepsInvExpFree  
			& $2.22 \times 10^{-6}$ 
			& $2.24 \times 10^{-6}$ 
			& $1.58 \times 10^{-5}$
			& $3.11 \times 10^{-5}$
			& $7.34 \times 10^{-3}$
			& $7.03 \times 10^{-3}$ 
			& $1.63 \times 10^{-1}$ 
			& $8.22 \times 10^{-1}$ \\
			
			PolyACSigExpFree   
			& $4.31 \times 10^{-6}$ 
			& $4.35 \times 10^{-6}$ 
			& $3.95 \times 10^{-5}$ 
			& $5.95 \times 10^{-5}$ 
			& $1.33 \times 10^{-2}$
			& $1.27 \times 10^{-2}$
			& $2.14 \times 10^{-1}$ 
			& $9.13 \times 10^{-1}$ \\
			
			\bottomrule
		\end{tabular}
	\end{adjustbox}
	
	\caption{Average performance of the selected neural networks across $25$ independent runs on the dataset $(A,B)\in [0,3]\times[10^{-7},1.22]$ (trained with the MSRE loss function minimised using mini-batch Adam with adaptive learning-rate.)}
	\label{tab:loss_sgd_var_DS2medium}
\end{table}

	The pointwise error distributions shown in Figure~\ref{fig:heatmap_AB_MSRE_DS2medium_0_3_1000_1em7_1p22_1000} further support these observations. The standard neural network again exhibits a pronounced region of elevated relative error in the low-volatility regime. In contrast, the asymptotic architectures maintain a significantly more homogeneous error profile throughout the domain.
	
	\begin{figure}[h!]
	\includegraphics[width=\linewidth]{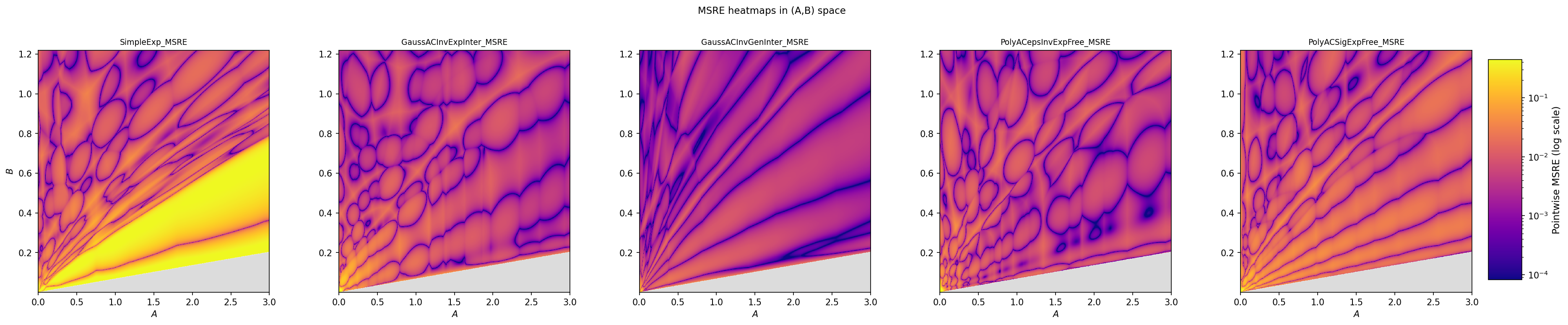}
	\centering
	\caption{Pointwise relative error, on a logarithmic scale, of the chosen architectures over the domain $(A,B)\in [0,3]\times[10^{-7},1.22]$. The evaluation grid contains $1000\times1000$ points.}
	\label{fig:heatmap_AB_MSRE_DS2medium_0_3_1000_1em7_1p22_1000}
	\end{figure}
	
	We next study the effect of Householder refinement. Figure~\ref{fig:log_convergence_Householder_2_iter_DS2medium_A_1p5} illustrates the convergence behaviour for a fixed log-moneyness value $A=1.5$ over the interval $B\in[10^{-7},1.22]$. The behaviour closely mirrors that observed on the larger benchmark. The asymptotic architectures already provide highly accurate initial guesses, allowing the Householder iterations to converge essentially to the maximum precision after only two iterations. By comparison, the standard neural network starts considerably farther from the solution, particularly in the low-volatility regime. This observation is quantified in Table~\ref{tab:stats_Householder_2_iter_DS2medium_A_1p5}. After two iterations, all asymptotic architectures achieve extremely similar levels of accuracy, with average logarithmic relative errors around $-37$. This confirms that once the initial approximation lies within the basin of rapid convergence of Householder's method, the precise choice of architecture becomes less critical.

	\begin{figure}[h!]
	\includegraphics[width=\linewidth]{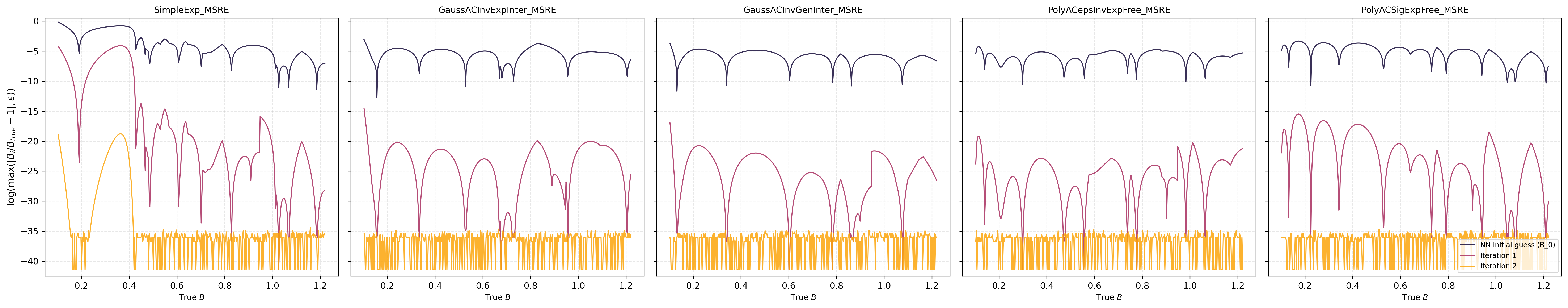}
	\centering
	\caption{Evolution of the logarithmic relative error \eqref{eq:LogRelError} for a fixed value $A=1.5$ over the interval $B\in [10^{-7},1.22]$. The curves correspond to the neural-network output ($B_0$), after one Householder iteration, and after two Householder iterations. $B$ is sampled uniformly over a grid of $500$ points.}
	\label{fig:log_convergence_Householder_2_iter_DS2medium_A_1p5}
	\end{figure}
	
	\begin{table}[ht]
  \centering
  \begin{tabular}{lccc}
    \hline
    Model & Avg Log Error & Std Log Error & Max Log Error \\
    \hline
    SimpleExp                & $-34.39$ & $6.00$ & $-18.78$ \\
    GaussACInvExpInter       & $-37.26$ & $2.42$ & $-34.65$ \\
    GaussACInvGenInter       & $-37.25$ & $2.43$ & $\mathbf{-34.94}$ \\
    PolyACepsInvExpFree      & $\mathbf{-37.33}$ & $2.42$ & $-34.79$ \\
    PolyACSigExpFree        & $-37.13$ & $\mathbf{2.34}$ & $-34.65$ \\
    \hline
  \end{tabular}
    \caption{Average, standard deviation and the maximum of the logarithm of the relative error \eqref{eq:LogRelError} for the chosen architectures after $2$ Householder iterations with $A=1.5$ and $B\in [10^{-7},1.22]$. $B$ is sampled uniformly over a grid of $500$ points.}
    \label{tab:stats_Householder_2_iter_DS2medium_A_1p5}
\end{table}

	Finally, Figure~\ref{fig:heatmap_log_rel_err_brent_DS2medium_0_3_1em7_1p22_200} compares the proposed neural-network-plus-Householder pipeline with Brent's method on the entire evaluation domain. The conclusions are consistent with those obtained on the larger benchmark. After two Householder iterations, all asymptotic architectures achieve a remarkably uniform level of accuracy across the parameter space. By contrast, Brent's method exhibits a much more irregular error landscape. These results further illustrate the practical value of combining a learned asymptotic approximation with a high-order local refinement scheme.

	\begin{figure}[h!]
	\includegraphics[width=\linewidth]{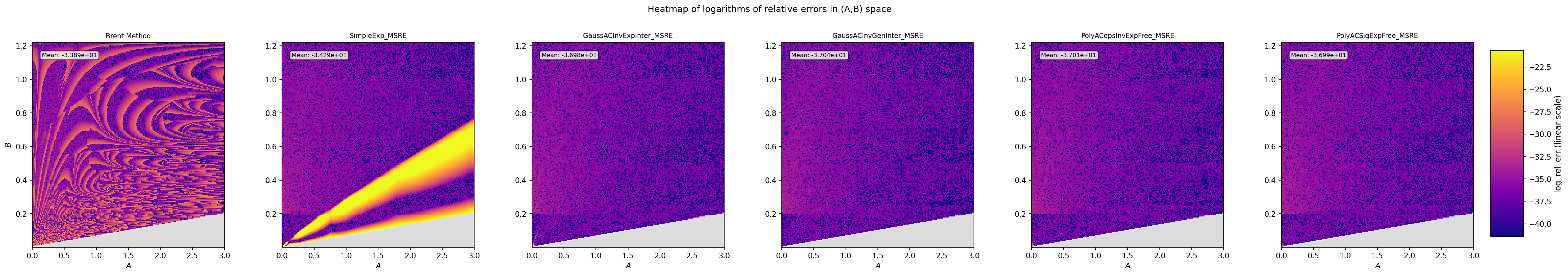}
	\centering
	\caption{Comparison of the logarithm of the relative error \eqref{eq:LogRelError} between Brent's method (maximum of $50$ iterations) and the proposed neural-network-plus-Householder approach on the domain $(A,B)\in [0,3]\times[10^{-7},1.22]$. The evaluation grid contains $200\times 200$ points.}
	\label{fig:heatmap_log_rel_err_brent_DS2medium_0_3_1em7_1p22_200}
	\end{figure}
	\subsubsection{A small-range dataset}
	We conclude our numerical study with a highly localised but particularly challenging dataset, $(A,B)\in [0,10^{-5}]\times [10^{-5},0.18]$. From a financial perspective, this regime corresponds to options that remain extremely close to the at-the-money configuration, while simultaneously including very small volatility values.
	
	The results are summarised in Figures~\ref{fig:evolution_loss_sgd_var_DS3small}, \ref{fig:heatmap_AB_MSRE_DS3small_0_1e5_1e5_0p18_500}, \ref{fig:log_convergence_Householder_2_iter_DS3small_A_5em6} and \ref{fig:heatmap_log_rel_err_brent_DS3small_0_1e5_1e5_0p18_200}, together with Tables~\ref{tab:loss_sgd_var_DS3small} and \ref{tab:stats_Householder_2_iter_DS3small_A_5em6}. 
	
	In contrast with the previous datasets, the performance gap between the standard neural network and the asymptotic architectures is noticeably smaller. This behaviour is unsurprising: over such a narrow interval of values of $A$, the inversion problem effectively becomes almost one-dimensional. As a consequence, even a standard feed-forward neural network can approximate the implied-volatility map reasonably well. Nevertheless, the proposed architectures continue to provide systematic improvements. In particular, the architecture PolyACepsInvExpInter clearly dominates all competitors across the majority of metrics. Compared with the standard neural network, it reduces the test MSRE from $9.04\times10^{-6}$ to $ 2.33\times10^{-7}$, corresponding to an improvement factor of approximately $39$. Similar gains are observed in the maximum absolute and relative residuals. While the improvement is less dramatic than in the previous datasets, it remains significant given the already small baseline errors. 
	
	An additional observation concerns the nature of the best-performing architectures. Unlike the medium-range dataset, where Gaussian AC gating functions were dominant, the strongest model now belongs to the polynomial family. Moreover, architectures depending only weakly on the variable $A$, such as PolyCSigExpInter and GaussCInvExpInter, become highly competitive. This is consistent with the fact that the entire dataset is concentrated near $A=0$. In this regime, the asymptotic partition derived in Section~\ref{sec:BSAnalysis} varies very little as a function of log-moneyness, reducing the importance of sophisticated $A$-dependent gating mechanisms.
	
	\begin{figure}[h!]
	\includegraphics[scale=0.4]{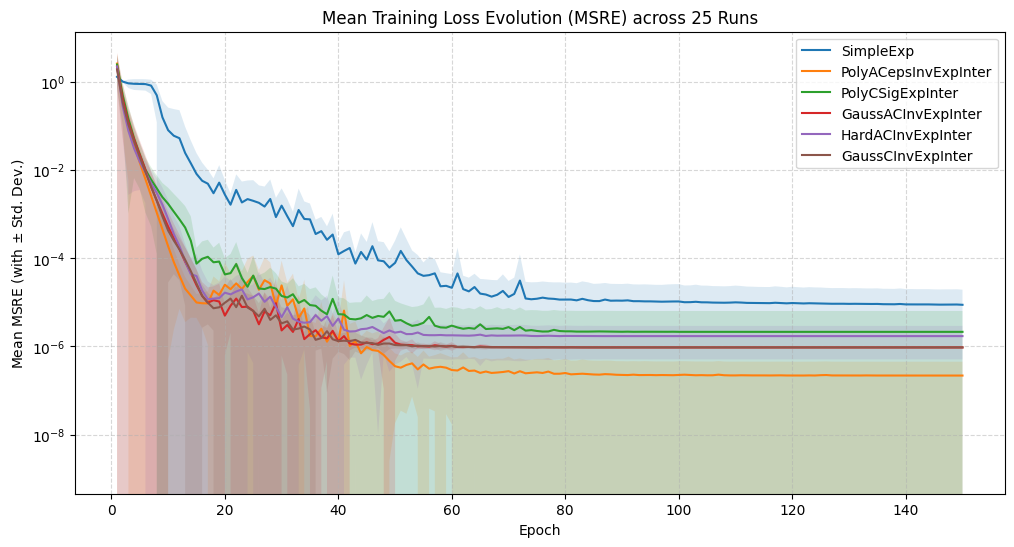} 
	\centering
	\caption{Evolution of the training MSRE averaged over $25$ independent runs for the selected architectures on the dataset $(A,B)\in [0,10^{-5}]\times [10^{-5},0.18]$. The shaded regions represent one standard deviation around the mean.}
	\label{fig:evolution_loss_sgd_var_DS3small}
	\end{figure}	
	
\begin{table}[ht]
	\centering
	\scriptsize
	\setlength{\tabcolsep}{3pt}
	\renewcommand{\arraystretch}{1.15}
	
	\begin{adjustbox}{max width=\textwidth}
		\begin{tabular}{lcccccccc}
			\toprule
			Model 
			& \makecell{$\mathrm{MSE}_\mathrm{train}$} 
			& \makecell{$\mathrm{MSE}_\mathrm{test}$} 
			& \makecell{$\mathrm{MSRE}_\mathrm{train}$} 
			& \makecell{$\mathrm{MSRE}_\mathrm{test}$} 
			& \makecell{$\Delta_\mathrm{train}$} 
			& \makecell{$\Delta_\mathrm{test}$} 
			& \makecell{$\delta_\mathrm{train}$} 
			& \makecell{$\delta_\mathrm{test}$} \\
			\midrule
			
			SimpleExp 
			& $1.25 \times 10^{-8}$ 
			& $1.25 \times 10^{-8}$ 
			& $8.80 \times 10^{-6}$ 
			& $9.04 \times 10^{-6}$ 
			& $3.61 \times 10^{-4}$
			& $3.61 \times 10^{-4}$
			& $3.27 \times 10^{-2}$ 
			& $3.25 \times 10^{-2}$ \\
			
			PolyACepsInvExpInter 
			& $5.05 \times 10^{-10}$
			& $5.10 \times 10^{-10}$ 
			& $\mathbf{2.19 \times 10^{-7}}$
			& $\mathbf{2.33 \times 10^{-7}}$
			& $\mathbf{6.24 \times 10^{-5}}$ 
			& $\mathbf{6.23 \times 10^{-5}}$
			& $\mathbf{9.62 \times 10^{-3}}$
			& $\mathbf{9.51 \times 10^{-3}}$ \\
			
			PolyCSigExpInter 
			& $6.85 \times 10^{-9}$
			& $6.90 \times 10^{-9}$ 
			& $2.14 \times 10^{-6}$ 
			& $2.20 \times 10^{-6}$ 
			& $1.47 \times 10^{-4}$ 
			& $1.47 \times 10^{-4}$ 
			& $3.49 \times 10^{-2}$ 
			& $3.44 \times 10^{-2}$ \\
			
			GaussACInvExpInter  
			& $5.18 \times 10^{-10}$
			& $5.23 \times 10^{-10}$ 
			& $9.52 \times 10^{-7}$
			& $9.42 \times 10^{-7}$ 
			& $8.07 \times 10^{-5}$
			& $8.06 \times 10^{-5}$
			& $3.27 \times 10^{-2}$ 
			& $3.20 \times 10^{-2}$ \\
			
			HardACInvExpInter   
			& $1.98 \times 10^{-9}$ 
			& $2.00 \times 10^{-9}$
			& $1.71 \times 10^{-6}$
			& $1.73 \times 10^{-6}$ 
			& $1.12 \times 10^{-4}$
			& $1.12 \times 10^{-4}$ 
			& $4.22 \times 10^{-2}$ 
			& $4.15 \times 10^{-2}$ \\
			
			GaussCInvExpInter   
			& $\mathbf{4.77 \times 10^{-10}}$ 
			& $\mathbf{4.81 \times 10^{-10}}$ 
			& $9.52 \times 10^{-7}$
			& $9.41 \times 10^{-7}$
			& $7.57 \times 10^{-5}$
			& $7.56 \times 10^{-5}$ 
			& $3.27 \times 10^{-2}$ 
			& $3.20 \times 10^{-2}$ \\
			\bottomrule
		\end{tabular}
	\end{adjustbox}
	\caption{Average performance of the selected neural networks across $25$ independent runs on the dataset $(A,B)\in [0,10^{-5}]\times [10^{-5},0.18]$ (trained with the MSRE loss function minimised using mini-batch Adam with adaptive learning-rate.)}
	\label{tab:loss_sgd_var_DS3small}
\end{table}

	Unlike the heat maps of the previous datasets, the error patterns in Figure~\ref{fig:heatmap_AB_MSRE_DS3small_0_1e5_1e5_0p18_500} are almost entirely horizontal. This observation confirms that the dependence of the inverse map on the log-moneyness is negligible over the considered range and that the approximation problem is driven primarily by the price variable. The fact that the proposed architectures continue to outperform the standard neural network even in this near one-dimensional setting demonstrates that their effectiveness is not solely due to the adaptive partitioning of the $A$-domain.

	\begin{figure}[h!]
	\includegraphics[width=\linewidth]{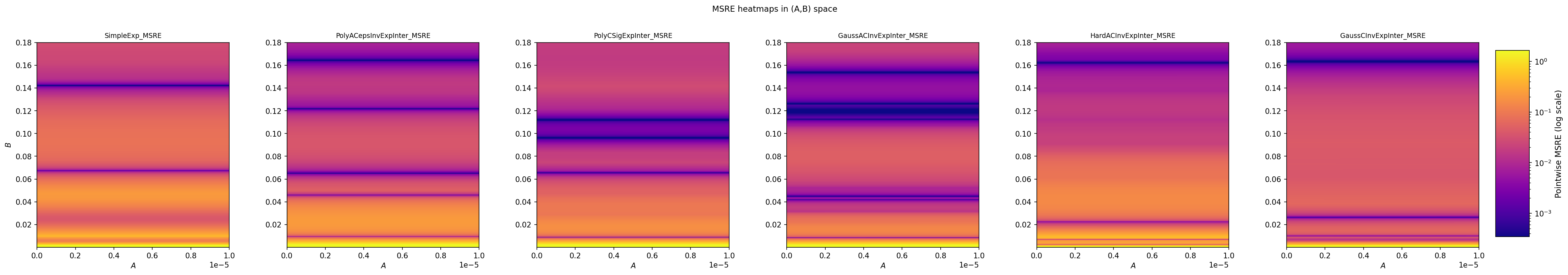}
	\centering
	\caption{Pointwise relative error, on a logarithmic scale, of the chosen architectures over the domain $(A,B)\in [0,10^{-5}]\times [10^{-5},0.18]$. The evaluation grid contains $500\times500$ points.}
	\label{fig:heatmap_AB_MSRE_DS3small_0_1e5_1e5_0p18_500}
	\end{figure}
	
	As in the previous experiments, we use the neural-network predictions as initial guesses for the third-order Householder scheme. Figure~\ref{fig:log_convergence_Householder_2_iter_DS3small_A_5em6} shows the resulting convergence behaviour for $A=5.10^{-6}$ and $B\in [10^{-5},0.18]$. All architectures, including the standard feed-forward network, lie sufficiently close to the exact solution for the Householder iterations to converge rapidly. Nevertheless, the asymptotic architectures consistently provide more accurate initial guesses and therefore achieve slightly better accuracy after refinement. 	The corresponding statistics are reported in Table~\ref{tab:stats_Householder_2_iter_DS3small_A_5em6}. The differences between neural-network architectures become relatively small after refinement, indicating that all considered models provide initial approximations lying well within the convergence region of the Householder method. 
	
	\begin{figure}[h!]
	\includegraphics[width=\linewidth]{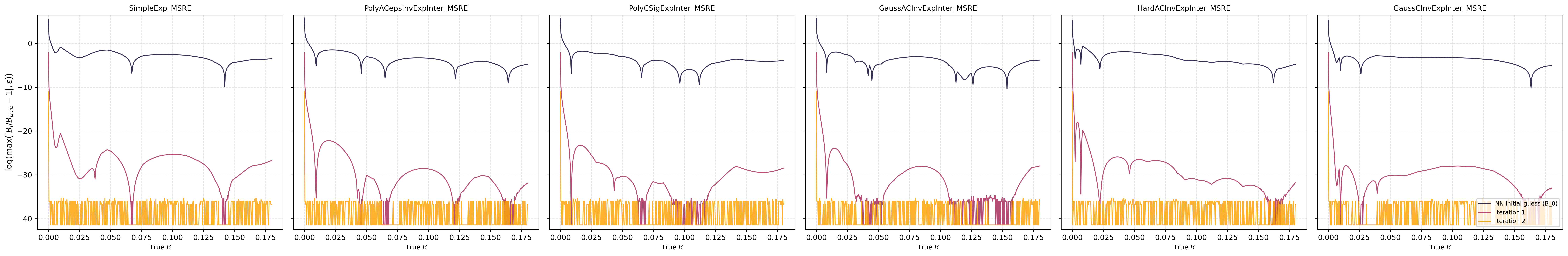}
	\centering
	\caption{Evolution of the logarithmic relative error \eqref{eq:LogRelError} for a fixed value $A=5.10^{-6}$ over the interval $B\in [10^{-5},0.18]$. The curves correspond to the neural-network output ($B_0$), after one Householder iteration, and after two Householder iterations. $B$ is sampled uniformly over a grid of $500$ points.}
	\label{fig:log_convergence_Householder_2_iter_DS3small_A_5em6}
	\end{figure}
	
	\begin{table}[htbp]
  \centering
  \begin{tabular}{lccc}
    \hline
    Model & Avg Log Error & Std Log Error & Max Log Error \\
    \hline
    SimpleExp            & $-38.70$ & $2.92$ & $-10.948$ \\
    PolyACepsInvExpInter     & $-39.06$ & $2.88$ & $-10.925$ \\
    PolyCSigExpInter         & $-39.07$ & $2.92$ & $-10.926$ \\
    GaussACInvExpInter       & $\mathbf{-39.40}$ & $\mathbf{2.85}$ & $-10.932$ \\
    HardACInvExpInter        & $-38.62$ & $2.91$ & $\mathbf{-10.957}$ \\
    GaussCInvExpInter        & $-38.94$ & $2.91$ & $-10.953$ \\
    \hline
  \end{tabular}
    \caption{Average, standard deviation and the maximum of the logarithm of the relative error \eqref{eq:LogRelError} for the chosen architectures after $2$ Householder iterations with $A=5.10^{-6}$ and $B\in [10^{-5},0.18]$. $B$ is sampled uniformly over a grid of $500$ points.}
    \label{tab:stats_Householder_2_iter_DS3small_A_5em6}
\end{table}

	Finally, Figure~\ref{fig:heatmap_log_rel_err_brent_DS3small_0_1e5_1e5_0p18_200} compares the proposed methodology with Brent's method. The most striking feature of this experiment is the behaviour of Brent's algorithm. Unlike in the previous datasets, where it could be argued that Brent's method remained somewhat competitive on average, it now performs poorly across a substantial fraction of the parameter space. The deterioration is visible almost uniformly throughout the domain and highlights the numerical difficulties associated with implied-volatility inversion in very low-volatility or low-log-moneyness regimes. By contrast, the neural-network-plus-Householder pipeline remains stable and highly accurate over the entire grid. These results suggest that the proposed methodology is not merely competitive with classical root-finding methods, but may in fact be more reliable in some of the most numerically challenging regions of the Black-Scholes parameter space.

	\begin{figure}[h!]
	\includegraphics[width=\linewidth]{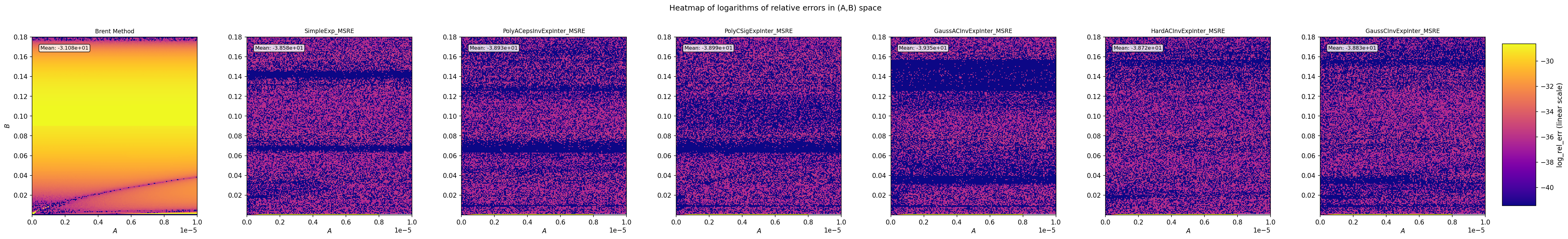}
	\centering
	\caption{Comparison of the logarithm of the relative error \eqref{eq:LogRelError} between Brent's method (maximum of $50$ iterations) and the proposed neural-network-plus-Householder approach on the domain $(A,B)\in [0,10^{-5}]\times [10^{-5},0.18]$. The evaluation grid contains $200\times 200$ points.}
	\label{fig:heatmap_log_rel_err_brent_DS3small_0_1e5_1e5_0p18_200}
	\end{figure}

\section{Conclusion and possible improvements}\label{sec:Conclusion}
	In this paper, we introduced a new family of neural-network architectures for the computation of Black-Scholes implied volatilities. The proposed approach is motivated by a detailed analysis of the Black-Scholes pricing function and, in particular, by the existence of distinct asymptotic and central regimes in the implied-volatility inversion problem. This novel construction, based on the decomposition \eqref{eq:IVGenFormula} and incorporating trainable gating mechanisms and specialised local approximations, embeds analytical knowledge of the Black-Scholes formula directly into the architecture while remaining fully data-driven. Following the seminal work \cite{Jackel}, we further refined the obtained neural approximations with a third-order Householder root-finding scheme. Thus, the resulting methodology combines the speed of machine-learning approximations with the accuracy of high-order iterative solvers.
	
	The numerical experiments demonstrate that the proposed architectures consistently outperform standard feed-forward neural networks across a broad range of parameter domains. We observe improvements of several orders of magnitude in both mean and worst-case relative errors, despite the use of substantially smaller networks. Moreover, the proposed architectures generalise remarkably well, with training and testing performances remaining nearly identical throughout the experiments.
	
	One of the most notable findings of this work is that the outputs generated by the proposed networks are already sufficiently accurate to place the Householder algorithm directly within its region of rapid convergence. Across all datasets considered, the proposed architectures reduced the test MSRE by between one and four orders of magnitude relative to standard feed-forward neural networks while requiring only two Householder iterations to attain accuracies close to machine precision. The neural-network component may therefore not necessarily be viewed as a replacement for classical root-finding techniques, but rather as an efficient mechanism for producing exceptionally good initial guesses.
	
	The numerical study also provides practical guidance regarding architectural choices. Among all the considered models, architectures belonging to the Gaussian family of gating functions generally achieved the best balance between accuracy, robustness and consistency across datasets. In particular, the GaussAC architectures proved especially effective whenever the range of log-moneyness values was sufficiently large for the geometry of the asymptotic regions to vary significantly with $A$. Consequently, if a single architecture had to be recommended as a default choice, our experiments suggest starting with a Gaussian AC architecture combined with the interpolation formulation \eqref{eq:InterFormula}.
	
	The results also highlight the importance of incorporating problem-specific inductive biases into neural-network design. While standard universal approximation results guarantee that sufficiently large feed-forward networks can approximate the implied-volatility function arbitrarily well, the proposed architectures achieve substantially better accuracy with significantly fewer parameters by exploiting the known analytical structure of the problem.

	Several directions for future work naturally emerge from this study. For instance, the local approximations $g_0$ and $g_1$ could be enriched. The asymptotic regions of the implied-volatility function possess well-understood qualitative properties and characteristic shapes that are not fully exploited by the generic neural networks used in this work. One promising direction would be to replace the current local approximations by parametrisations that explicitly reproduce the known asymptotic behaviour of the implied volatility in the limits $C\to0^+$ and $C\to1^-$. 
	
	Finally, while the present work focuses exclusively on the inversion of the Black-Scholes formula, the underlying philosophy is considerably more general. Many calibration problems in quantitative finance involve inverse mappings exhibiting multiple regimes, strong non-linearities and difficult asymptotic behaviour. Examples include local-volatility, stochastic-volatility and rough-volatility models. Extending the proposed regime-aware architectures to such settings appears particularly promising.
	
	\bibliographystyle{plain}
	\bibliography{references}

@article{AKP,
  title={Neural networks with asymptotics control},
  author={Antonov, Alexandre and Konikov, Michael and Piterbarg, Vladimir},
  journal={Available at SSRN 3544698},
  year={2020}
}

@article{AL,
  title={Asymptotics for exponential L{\'e}vy processes and their volatility smile: survey and new results},
  author={Andersen, Leif and Lipton, Alexander},
  journal={International Journal of Theoretical and Applied Finance},
  volume={16},
  number={01},
  pages={1350001},
  year={2013},
  publisher={World Scientific}
}

@book{AS,
     TITLE = {Handbook of mathematical functions with formulas, graphs, and
              mathematical tables},
    SERIES = {A Wiley-Interscience Publication},
    EDITOR = {Abramowitz, Milton and Stegun, Irene A.},
      NOTE = {Reprint of the 1972 edition,
              Selected Government Publications},
 PUBLISHER = {John Wiley \& Sons, Inc., New York; John Wiley \& Sons, Inc.,
              New York},
      YEAR = {1984},
     PAGES = {xiv+1046},
      ISBN = {0-471-80007-4},
   MRCLASS = {00A20 (00A22 65-00)},
  MRNUMBER = {757537},
}

@article{ATV,
  title={Deep smoothing of the implied volatility surface},
  author={Ackerer, Damien and Tagasovska, Natasa and Vatter, Thibault},
  journal={Advances in Neural Information Processing Systems},
  volume={33},
  pages={11552--11563},
  year={2020}
}

@article{Brent,
  title={An algorithm with guaranteed convergence for finding a zero of a function},
  author={Brent, Richard P.},
  journal={The computer journal},
  volume={14},
  number={4},
  pages={422--425},
  year={1971},
  publisher={Oxford University Press}
}

@article{BCS,
  title={Computing the Black-Scholes implied volatility: Generalization of a simple formula},
  author={Bharadia, MAJ and Christofides, N and Salkin, GR},
  journal={Advances in futures and options research},
  volume={8},
  pages={15--30},
  year={1995},
  publisher={JAI PRESS INC.}
}

@article{BCWR,
  title={Variational Bayesian mixture of experts models and sensitivity analysis for nonlinear dynamical systems},
  author={Baldacchino, Tara and Cross, Elizabeth J and Worden, Keith and Rowson, Jennifer},
  journal={Mechanical Systems and Signal Processing},
  volume={66},
  pages={178--200},
  year={2016},
  publisher={Elsevier}
}

@article{BS,
  title={The pricing of options and corporate liabilities},
  author={Black, Fischer and Scholes, Myron},
  journal={Journal of political economy},
  volume={81},
  number={3},
  pages={637--654},
  year={1973},
  publisher={The University of Chicago Press}
}

@article{BS2,
  title={A simple formula to compute the implied standard deviation},
  author={Brenner, Menachem and Subrahmanyan, Marti G},
  journal={Financial Analysts Journal},
  volume={44},
  number={5},
  pages={80--83},
  year={1988},
  publisher={Taylor \& Francis}
}

@article{Chance,
  title={A generalized simple formula to compute the implied volatility},
  author={Chance, Don M},
  journal={Financial Review},
  volume={31},
  number={4},
  pages={859--867},
  year={1996},
  publisher={Wiley Online Library}
}

@article{Cybenko,
  title={Approximation by superpositions of a sigmoidal function},
  author={Cybenko, George},
  journal={Mathematics of control, signals and systems},
  volume={2},
  number={4},
  pages={303--314},
  year={1989},
  publisher={Springer}
}

@article{CM,
  title={A note on a simple, accurate formula to compute implied standard deviations},
  author={Corrado, Charles J and Miller Jr, Thomas W},
  journal={Journal of Banking \& Finance},
  volume={20},
  number={3},
  pages={595--603},
  year={1996},
  publisher={Elsevier}
}

@article{CN,
  title={An improved approach to computing implied volatility},
  author={Chambers, Donald R and Nawalkha, Sanjay K},
  journal={Financial Review},
  volume={36},
  number={3},
  pages={89--100},
  year={2001},
  publisher={Wiley Online Library}
}

@article{DG,
  title={Shape preserving piecewise rational interpolation},
  author={Delbourgo, Roger and Gregory, John A},
  journal={SIAM journal on scientific and statistical computing},
  volume={6},
  number={4},
  pages={967--976},
  year={1985},
  publisher={SIAM}
}

@article{Dupire,
  title={Pricing with a smile},
  author={Dupire, Bruno and others},
  journal={Risk},
  volume={7},
  number={1},
  pages={18--20},
  year={1994}
}

@article{DVPP,
  title={Machine learning for option pricing: an empirical investigation of network architectures},
  author={Della Corte, Serena and Van Mieghem, Laurens and Papapantoleon, Antonis and Papazoglou-Hennig, Jonas},
  journal={The European Journal of Finance},
  pages={1--31},
  year={2026},
  publisher={Taylor \& Francis}
}

@misc{FG,
      title={Deeply Learning Derivatives}, 
      author={Ryan Ferguson and Andrew Green},
      year={2018},
      eprint={1809.02233},
      archivePrefix={arXiv},
      primaryClass={q-fin.CP},
      url={https://arxiv.org/abs/1809.02233}, 
}

@article{GHMP,
	title={The Chebyshev method for the implied volatility},
	author={Glau, Kathrin and Herold, Paul and Madan, Dilip B and P{\"o}tz, Christian},
	journal={arXiv preprint arXiv:1710.01797},
	year={2017}
}

@article{GG,
  title={Pricing and hedging derivative securities with neural networks and a homogeneity hint},
  author={Garcia, Ren{\'e} and Gen{\c{c}}ay, Ramazan},
  journal={Journal of Econometrics},
  volume={94},
  number={1-2},
  pages={93--115},
  year={2000},
  publisher={Elsevier}
}

@article{GU,
  title={Approximating {M}ills ratio},
  author={Gasull, Armengol and Utzet, Frederic},
  journal={Journal of Mathematical Analysis and Applications},
  volume={420},
  number={2},
  pages={1832--1853},
  year={2014},
  publisher={Elsevier}
}

@article{Hernandez,
	title={Model calibration with neural networks},
	author={Hernandez, Andres},
	journal={Available at SSRN 2812140},
	year={2016}
}

@article{HLP,
  title={A nonparametric approach to pricing and hedging derivative securities via learning networks},
  author={Hutchinson, James M and Lo, Andrew W and Poggio, Tomaso},
  journal={The journal of Finance},
  volume={49},
  number={3},
  pages={851--889},
  year={1994},
  publisher={Wiley Online Library}
}

@article{HMT,
  title={Deep learning volatility: a deep neural network perspective on pricing and calibration in (rough) volatility models},
  author={Horvath, Blanka and Muguruza, Aitor and Tomas, Mehdi},
  journal={Quantitative Finance},
  volume={21},
  number={1},
  pages={11--27},
  year={2021},
  publisher={Taylor \& Francis}
}

@article{Hornik,
  title={Approximation capabilities of multilayer feedforward networks},
  author={Hornik, Kurt},
  journal={Neural networks},
  volume={4},
  number={2},
  pages={251--257},
  year={1991},
  publisher={Elsevier}
}

@inproceedings{Householder,
  title={The numerical treatment of a single nonlinear equation},
  author={Alston S. Householder},
  year={1970},
  url={https://api.semanticscholar.org/CorpusID:117523261}
}

@inproceedings{HZRS,
  title={Deep residual learning for image recognition},
  author={He, Kaiming and Zhang, Xiangyu and Ren, Shaoqing and Sun, Jian},
  booktitle={Proceedings of the IEEE conference on computer vision and pattern recognition},
  pages={770--778},
  year={2016}
}

@article{Jackel,
	title={Let's be rational},
	author={J{\"a}ckel, Peter},
	journal={Wilmott},
	volume={2015},
	number={75},
	pages={40--53},
	year={2015},
	publisher={Wiley Online Library}
}

@article{JJNH,
  title={Adaptive mixtures of local experts},
  author={Jacobs, Robert A and Jordan, Michael I and Nowlan, Steven J and Hinton, Geoffrey E},
  journal={Neural computation},
  volume={3},
  number={1},
  pages={79--87},
  year={1991},
  publisher={MIT Press}
}

@article{Li,
  title={Approximate inversion of the Black--Scholes formula using rational functions},
  author={Li, Minqiang},
  journal={European Journal of Operational Research},
  volume={185},
  number={2},
  pages={743--759},
  year={2008},
  publisher={Elsevier}
}

@article{LOB,
  title={Pricing options and computing implied volatilities using neural networks},
  author={Liu, Shuaiqiang and Oosterlee, Cornelis W and Bohte, Sander M},
  journal={Risks},
  volume={7},
  number={1},
  pages={16},
  year={2019},
  publisher={MDPI}
}

@article{Merton,
  title={Theory of rational option pricing},
  author={Merton, Robert C and others},
  year={1971},
  publisher={World Scientific}
}

@article{Mills,
  title={Table of the ratio: area to bounding ordinate, for any portion of normal curve},
  author={Mills, John P},
  journal={Biometrika},
  pages={395--400},
  year={1926},
  publisher={JSTOR}
}

@book{Rokach,
  title={Pattern classification using ensemble methods},
  author={Rokach, Lior},
  volume={75},
  year={2010},
  publisher={World Scientific}
}

@article{Tehranchi,
  title={Asymptotics of implied volatility far from maturity},
  author={Tehranchi, Michael R},
  journal={Journal of Applied Probability},
  volume={46},
  number={3},
  pages={629--650},
  year={2009},
  publisher={Cambridge University Press}
}

@article{Tehranchi2,
  title={Uniform bounds for Black--Scholes implied volatility},
  author={Tehranchi, Michael R},
  journal={SIAM Journal on Financial Mathematics},
  volume={7},
  number={1},
  pages={893--916},
  year={2016},
  publisher={SIAM}
}

@article{VC,
  title={VolGAN: a generative model for arbitrage-free implied volatility surfaces},
  author={Vuleti{\'c}, Milena and Cont, Rama},
  journal={Applied Mathematical Finance},
  volume={31},
  number={4},
  pages={203--238},
  year={2024},
  publisher={Taylor \& Francis}
}

@misc{SGS,
      title={Highway Networks}, 
      author={Rupesh Kumar Srivastava and Klaus Greff and Jürgen Schmidhuber},
      year={2015},
      eprint={1505.00387},
      archivePrefix={arXiv},
      primaryClass={cs.LG},
      url={https://arxiv.org/abs/1505.00387}, 
}

@article{MK,
	title={The calculation of implied variances from the Black-Scholes model: A note},
	author={Manaster, Steven and Koehler, Gary},
	journal={The Journal of Finance},
	volume={37},
	number={1},
	pages={227--230},
	year={1982},
	publisher={JSTOR}
}

@article{BHMST,
	title={On deep calibration of (rough) stochastic volatility models},
	author={Bayer, Christian and Horvath, Blanka and Muguruza, Aitor and Stemper, Benjamin and Tomas, Mehdi},
	journal={arXiv preprint arXiv:1908.08806},
	year={2019}
}

@book{Gulisashvili,
	title={Analytically tractable stochastic stock price models},
	author={Gulisashvili, Archil},
	year={2012},
	publisher={Springer Science \& Business Media}
}

@article{CKT,
	title={A generative adversarial network approach to calibration of local stochastic volatility models},
	author={Cuchiero, Christa and Khosrawi, Wahid and Teichmann, Josef},
	journal={Risks},
	volume={8},
	number={4},
	pages={101},
	year={2020},
	publisher={MDPI}
}

@article{RW,
  title={Neural networks for option pricing and hedging: a literature review},
  author={Ruf, Johannes and Wang, Weiguan},
  journal={Journal of Computational Finance},
  volume={24},
  number={1},
  pages={1--46},
  year={2020}
}

@article{SS,
  title={DGM: A deep learning algorithm for solving partial differential equations},
  author={Sirignano, Justin and Spiliopoulos, Konstantinos},
  journal={Journal of computational physics},
  volume={375},
  pages={1339--1364},
  year={2018},
  publisher={Elsevier}
}

@article{SS2,
  title={Stochastic gradient descent in continuous time},
  author={Sirignano, Justin and Spiliopoulos, Konstantinos},
  journal={SIAM Journal on Financial Mathematics},
  volume={8},
  number={1},
  pages={933--961},
  year={2017},
  publisher={SIAM}
}

@article{YT,
  title={Monotonicity results and new bounds for the {M}ills ratio},
  author={Yang, Zhen-Hang and Tian, Jing-Feng},
  journal={Statistical Papers},
  volume={66},
  number={1},
  pages={25},
  year={2025},
  publisher={Springer}
}

@article{NZW,
	title={Computing Volatility Surfaces using Generative Adversarial Networks with Minimal Arbitrage Violations},
	author={Na, Andrew and Zhang, Meixin and Wan, Justin},
	journal={arXiv preprint arXiv:2304.13128},
	year={2023}
}

\appendix
	\section{Accurate implementation of the Black-Scholes pricing formula}\label{ap:AsymptoticBS}
	The direct implementation of the formula \eqref{eq:BSpriceCompact} may suffer from several well-known numerical difficulties, including underflow, overflow and, most importantly, catastrophic subtractive cancellation. Such phenomena become particularly pronounced when the total volatility is small. In this regime, the price is obtained as the difference between two nearly equal quantities, resulting in a potentially large loss of relative accuracy. To overcome this difficulty, we follow an approach inspired by \cite{Jackel} and reformulate the pricing function in terms of the Gaussian Mills ratio. This representation isolates the numerical challenges into a single difference term for which accurate Taylor and asymptotic expansions can be derived.
	
	We define the function $Z$ by
	\[
	Z(t)
	= \frac{\Phi(-t)}{\varphi(t)},
	\]
	where $\Phi$ and $\varphi$ denote the cumulative distribution function and density function of a standard Gaussian random variable, respectively. $Z$ is known as the Gaussian Mills ratio, a widely studied object in probability (e.g., see \cite{Mills, GU, YT}.) With this definition, and introducing the variables $h=\frac{A}{B}$ and $t=\frac{B}{2}$, the normalised Black-Scholes price may be written as
	\begin{equation}\label{eq:BSpriceMills}
	C_{\mathrm{BS}}(A,B)
	=\frac{1}{\sqrt{2\pi}}
	\exp\left(-\frac{1}{2}\left(h-t \right)^2 \right)
	\left(Z\left(h-t \right)-Z\left(h+t \right)\right).
	\end{equation}
	
	Compared with \eqref{eq:BSpriceCompact}, the representation \eqref{eq:BSpriceMills} isolates the loss of precision into the computation of the difference $Z(h-t)-Z(h+t)$. When $t$ is small, this difference can be evaluated accurately by means of a Taylor expansion around $h$, assuming that $Z$ and its derivatives at $h$ can be computed precisely or are not too small, which is the case when $h$ is relatively small. On the other hand, when $h$ is large, the Mills ratio and its derivatives become extremely small, making an asymptotic expansion more appropriate. The remainder of this appendix develops both approximations and explains how they are combined in practice.
	\subsection{A Taylor expansion of the call price in the small-$t$ regime}
	When $t$ is small, the quantity $Z(h-t)-Z(h+t)$ is affected by subtractive cancellation. To preserve relative accuracy, we replace the direct evaluation of this difference by a Taylor expansion of $Z$ around $h$. Expanding up to order $N$, we write
	\[
	Z(h+t)=\sum_{k=0}^{N}Z^{(k)}(h)\frac{t^k}{k!}+ Z^{(N+1)}(h+\theta_t)\frac{t^{N+1}}{(N+1)!},
	\]
	where $\theta_t$ denotes a number strictly between $0$ and $t$. Hence
	\[
	Z(h+t)-Z(h-t)=\sum_{0\leq 2k+1 \leq N}2Z^{(2k+1)}(h)\frac{t^{2k+1}}{(2k+1)!}
				+\left(Z^{(N+1)}(h+\theta_t)+(-1)^{N}Z^{(N+1)}(h+\theta_{-t})\right)\frac{t^{N+1}}{(N+1)!}.
	\]
	Denote by $J_NZ(h,t)$ the above polynomial Taylor approximation of the difference $Z(h+t)-Z(h-t)$
	\[
	J_NZ(h,t):= \sum_{0\leq 2k+1 \leq N}2Z^{(2k+1)}(h)\frac{t^{2k+1}}{(2k+1)!}.
	\]
	The derivatives appearing above can be easily computed. In fact, $Z$ satisfies
	\[
	\forall x\colon\quad 
	Z'(x)=-1+xZ(x).
	\]
	Hence, the sequence $Z^{(n)}(x)$ satisfies the following recurrence
	\[
	Z^{(n+1)}(x)= x Z^{(n)}(x)+ n Z^{(n-1)}(x).
	\]
	This recurrence allows all derivatives of $Z$ to be computed recursively once the value of $Z(x)$ is known. In practice, the latter can be evaluated with high relative accuracy through the scaled complementary error function $\erfcx$ as follows
	\[
	Z(x)=\sqrt{\frac{\pi}{2}} \erfcx\left(\frac{x}{\sqrt{2}} \right).
	\]
	We do not attempt a detailed study of the analytical properties of the Mills ratio, as an extensive literature already exists on the subject. Instead, we derive a simple error estimate sufficient for selecting the truncation order used in our implementation. For this purpose, it is easy to show that (bound valid for $x\geq 0$ only)
	\begin{equation}\label{eq:MillsDerivativeBound}
	(-1)^nZ^{(n)}(x)
	=e^{\frac{x^2}{2}}\int_0^{\infty}s^ne^{-\frac{(s+x)^2}{2}}\drm s
	\leq e^{\frac{x^2}{2}}\int_0^{\infty}s^ne^{-\frac{s^2}{2}}\drm s
	= e^{\frac{x^2}{2}} 2^{\frac{n-1}{2}}\Gamma\left(\frac{n+1}{2} \right),	
	\end{equation}
	with $\Gamma$ being the classical Gamma function. Hence $|Z^{(n)}(x)| \leq e^{\frac{x^2}{2}} 2^{\frac{n-1}{2}}\Gamma\left(\frac{n+1}{2} \right)$. We recall that, for $n=2p+1$ odd, then 
	\[\Gamma\left(\frac{n+1}{2} \right) = \Gamma\left(p+1 \right)= p!.\]
	If $|t|\leq t_{\min} \leq h$, then (with $N=2p$)
	\[
	\left|Z(h+t)-Z(h-t)-J_{2p}Z(h,t) \right|
	\leq 
	2^{p+1} \frac{p!}{(2p+1)!} e^{\frac{(t_{\min}+ h)^2}{2}} t_{\min}^{2p+1} 
	=e^{\frac{(t_{\min}+ h)^2}{2}} t_{\min}^{2p+1} \prod_{k=1}^p\frac{2}{p+k}.
	\]
	Recall that Inequality \eqref{eq:MillsDerivativeBound} is only sharp for small values of $x\geq 0$ ($x$ corresponding here to $h+\theta_t$ or $h+\theta_{-t}$.) Taking for instance $N=16$ and $t<0.1$, the error term is bounded from above by $\epsilon = 10^{-20}$ as long as $h<3.8023$. For high values of $h$, the bound we used above is not necessarily sharp.

	The preceding discussion suggests a natural range of validity for the Taylor approximation: $t$ should be sufficiently small to control the truncation error, while $h$ should remain moderate so that both $Z(h)$ and its derivatives can be evaluated accurately. In the implementation used throughout this paper, the Taylor expansion is employed whenever $h\leq 35$, a threshold beyond which the asymptotic expansion described below becomes more efficient and numerically reliable.
	\subsection{An asymptotic expansion of the call price in the large-$h$ regime}
	We recall the following classical approximation result.
	\begin{theo}{\cite[(26.2.12)]{AS}}\label{theo:MillsAsymptote} For $t>0$ and $N\in \Nbb^*$, one has
	\[
	\Phi(-t)
	= \frac{\varphi(t)}{t}\left(
	1+ \sum_{n=1}^N \frac{(-1)^n(2n-1)!!}{t^{2n}}
	\right) 
	+ R_N(t),
	\]
	with
	\[
	R_N(t) = (-1)^{N+1}(2N+1)!!\int_{t}^{\infty}\frac{\varphi(u)}{u^{2N+2}}\drm u.
	\]
	We have the trivial bound $|R_N(t)|\leq \frac{(2N+1)!!}{t^{2N+2}}$. 
	\end{theo}
	Applying Theorem \ref{theo:MillsAsymptote} to our case, we get the following expansion
	\begin{prop}\label{prop:AsympExpPrice} For $h>0$, $|t|<h$ and $N_1, N_2\in \Nbb^*$
	\[
	Z(h+t)- Z(h-t) = -2 \sum\limits_{n=0}^{N_1}  
			\sum\limits_{1\leq 2p+1\leq N_2} \frac{\gamma_{2p,n}}{h^{2n+1}}
			\left(\frac{t}{h}\right)^{2p+1} 
			+ \varepsilon_{1,N_1,N_2} + \varepsilon_{2,N_1}.
	\]
	where
		\[
	\forall n,k \geq 0\colon\quad
	\gamma_{k,n} =\left\{
	\begin{array}{ll}
	1&\text{if }n=0,\\
	(-1)^n(2n-1)!! {2n+k+1 \choose 2n}&\text{if }n\geq 1, \\
	\end{array}
	\right.
	\]
	and
	\[
	\varepsilon_{1,N_1,N_2} = \sum\limits_{n=0}^{N_1}  \frac{\gamma_{N_2,n}}{h^{2n+1}}
		\left(\frac{(-1)^{N_2+1}}{(1+\theta_{t/h})^{2n+N_2+2}}-\frac{1}{(1+\theta_{-t/h})^{2n+N_2+2}}\right)\left(\frac{t}{h}\right)^{N_2+1},
	\]
	where, $\theta_x$ is a real number strictly between $0$ and $x$, and
	\[
	\varepsilon_{2,N_1} = \widehat{R}_N(h+t)-\widehat{R}_N(h-t).
	\]
	with 
	\[
	\widehat{R}_N(t)
	= \frac{R_N(t)}{\varphi(t)}
	=(-1)^{N+1}(2N+1)!!\int_{t}^{\infty}\frac{\varphi(u)}{\varphi(t) u^{2N+2}}\drm u.
	\]
	\end{prop}	
	\begin{proof} 
	Using Theorem \ref{theo:MillsAsymptote}, we write, for $t>0$ and $N_1\in \Nbb^*$
	\[
	Z(h+t) 
	=\sum\limits_{n=0}^{N_1} \frac{\alpha_n}{(h+t)^{2n+1}}
	+ \widehat{R}_{N_1}(h+t)
	= \sum\limits_{n=0}^{N_1} \frac{\alpha_n}{h^{2n+1}} \frac{1}{(1+\frac{t}{h})^{2n+1}}
	+ \widehat{R}_{N_1}(h+t),\]
	where
	\[
	\alpha_n=\left\{
	\begin{array}{ll}
	1&\text{if }n=0,\\
	(-1)^n(2n-1)!!&\text{if }n\geq 1.\\
	\end{array}
	\right. 
	\]
	If we denote $f_n: x> -1 \mapsto \frac{1}{(1+x)^n}$, then its successive derivatives are given by
	\[
	\forall k\in \Nbb^*\colon,\;\; \forall x>0: \quad 
	f_n^{(k)}(x)=\frac{(-1)^k n(n+1)\cdots (n+k-1)}{(1+x)^{n+k}}.
	\]
	Taylor's expansion of $f_n$ around $0$ to the order $N_2\in \Nbb^*$ gives, for $x>-1$
	\[
	\frac{1}{(1+x)^n} = \sum_{k=0}^{N_2} \frac{(-1)^k n(n+1)\cdots (n+k-1)}{k!} x^k 
						+\frac{(-1)^{N_2+1} n(n+1)\cdots (n+N_2)}{(N_2+1)!(1+\theta_x)^{n+N_2+1}} x^{N_2+1},
	\]
	with $\theta_x$ denoting a number strictly between $0$ and $x$. This implies, for $|x|<1$
	\[\begin{array}{rcl}
	\frac{1}{(1+x)^n}-\frac{1}{(1-x)^n} 
	&=& -2\sum\limits_{1\leq 2p+1\leq N_2} \frac{n(n+1)\cdots (n+2p)}{(2p+1)!} x^{2p+1} \\
	&& +\frac{n(n+1)\cdots (n+N_2)}{(N_2+1)!} x^{N_2+1}
		(\frac{(-1)^{N_2+1}}{(1+\theta_x)^{n+N_2+1}}-\frac{1}{(1+\theta_{-x})^{n+N_2+1}}).	
	\end{array}
	\]
	In particular, this gives, for $n\in\inti{0}{N_1}$
	\[\begin{array}{rcl}
	\frac{1}{(1+x)^{2n+1}}-\frac{1}{(1-x)^{2n+1}} 
	&=& -2\sum\limits_{1\leq 2p+1\leq N_2} \frac{(2n+1)\cdots (2n+1+2p)}{(2p+1)!} x^{2p+1} \\
	&& +\frac{(2n+1)\cdots (2n+1+N_2)}{(N_2+1)!} x^{N_2+1}
		(\frac{(-1)^{N_2+1}}{(1+\theta_x)^{2n+N_2+2}}-\frac{1}{(1+\theta_{-x})^{2n+N_2+2}}).	
	\end{array}
	\]
	In summary, this gives, for $t<h$
	\[\begin{array}{rcl}
	Z(h+t)- Z(h-t)
	&=&-2 \sum\limits_{n=0}^{N_1} 
			\sum\limits_{1\leq 2p+1\leq N_2} \frac{\alpha_n}{h^{2n+1}}  \beta_{2p,n}
			(\frac{t}{h})^{2p+1}\\
	&& + \sum\limits_{n=0}^{N_1}  \frac{\alpha_n}{h^{2n+1}}
			\beta_{N_2,n}
		(\frac{(-1)^{N_2+1}}{(1+\theta_{t/h})^{2n+N_2+2}}-\frac{1}{(1+\theta_{-t/h })^{2n+N_2+2}})(\frac{t}{h})^{N_2+1}\\
	&& + \widehat{R}_N(h+t)-\widehat{R}_N(h-t),\\
	\end{array}
	\]
	where 
	\[
	\beta_{k,n}=\frac{(2n+1)\cdots (2n+1+k)}{(k+1)!}.
	\]
	The proof is now complete.
	\end{proof}
	We assume $|\frac{t}{h}|\leq \tau < 1$ and $h>h_\urm (=35)$. We have on the one hand
	\[
	|\varepsilon_{1,N_1,N_2}| 
	\leq 
	\sum\limits_{n=0}^{N_1}  \frac{|\gamma_{N_2,n}|}{h_\urm^{2n+1}}
	(1+\frac{1}{(1- \tau)^{2n+N_2+2}})
	\tau^{N_2+1}.
	\]
	Using the bound $|\widehat{R}_N(x)| \leq  \frac{(2N-1)!!}{t^{2N+1}}$, we get on the other hand
	\[
	|\varepsilon_{2,N_1}| 
	\leq 
	\frac{(2N_1-1)!!}{h_\urm^{2N_1+1}}\left(1+\frac{1}{(1- \tau)^{2N_1+1}}\right).
	\]
	Given an error threshold $\epsilon$, and in order to choose $N_1$ and $N_2$ in practice so that to ensure 
	\[
	|\varepsilon_{1,N_1,N_2}|+|\varepsilon_{2,N_1}|\leq \epsilon,
	\]
	we may first compute a value for $N_1$ such that
		\[
	4 \frac{(2N_1-1)!!}{h_\urm^{2N_1+1}\epsilon}<1,
	\]
	which ensures that we can find $\tau \in (0,1)$ so that $|\varepsilon_{2,N_1}| \leq \epsilon/2$. With this value of $N_1$ fixed, we can compute a maximal value for $\tau$. Finally, we may compute $N_2$ such that $|\varepsilon_{1,N_1,N_2}|\leq \epsilon/2$.\\
	
	For instance, with $\epsilon=10^{-20}$ and $h_\urm=35$, we can take $N_1\in \inti{9}{98}$ (with an overflow happening at the value $99$). Taking $N_1=17$ (as in \cite{Jackel}), we compute $\tau_{\max}=0.6154$ as an acceptable maximum value for $\tau$. Finally, taking $\tau_{\max}=0.1$, we may take $N_2\geq 20$ for example. 
	\subsection{Summary of the used approximations}\label{subsec:SummaryExpansions}	
	We summarise here the computation scheme in two steps:
	\begin{itemize}
	\item For $t<0.1$ and $h<35$, define (with $N=16$)
	\[
	\Delta Z(h,t)= -2 \sum_{1\leq 2k+1 \leq N}Z^{(2k+1)}(h)\frac{t^{2k+1}}{(2k+1)!},
	\]
	where $Z$ and its derivatives are computed recursively as follows
	\[
	Z(x)=\sqrt{\frac{\pi}{2}} \erfcx\left(\frac{x}{\sqrt{2}} \right),
	\quad
	Z'(x)=-1+xZ(x)
	\quad \text{and} \quad
	Z^{(n+1)}(x)= x Z^{(n)}(x)+ n Z^{(n-1)}(x).
	\]
	\item For $h\geq 35$ and $t<\tau_{\max} h$, and with $N_1=17$, $N_2=20$ and $\tau_{\max}=0.1$, define
	\[
	\Delta Z(h,t)=2\sum\limits_{n=0}^{N_1} 
			\sum\limits_{1\leq 2p+1\leq N_2} \frac{\gamma_{2p,n}}{h^{2n+1}}  
			\left(\frac{t}{h}\right)^{2p+1},
	\]
	with $\gamma_{2p,n}$ defined as in Proposition \ref{prop:AsympExpPrice}.
	\item Otherwise, we define
	\[
	\Delta Z(h,t)= Z(h-t)-Z(h+t),
	\]
	(with $Z$ computed using the $\erfcx$ function.)
	\end{itemize}
	Then approximate the price as follows
	\begin{itemize}
	\item If $h-t\leq 35$, use the approximation 
	\[
	C_{\mathrm{BS}}(A,B) = \frac{1}{\sqrt{2\pi}}
	\exp\left(-\frac{1}{2}\left(h-t \right)^2 \right) \Delta Z(h,t).
	\]
	\item Otherwise, compute as follows
		\[
	C_{\mathrm{BS}}(A,B) = \frac{1}{\sqrt{2\pi}}
	\exp\left(-\frac{1}{2}\left(h-t \right)^2 +\log \Delta Z(h,t)\right).
	\]
	\end{itemize}
	In Figures \ref{fig:TaylorOscillations_t_fixed} and \ref{fig:TaylorOscillations_h_fixed}, we show the difference in the evaluation of the price by naively using the formula \eqref{eq:BSpriceCompact} or one of the expansions given above.
	\begin{figure}[h!]
  \centering
  \begin{subfigure}{0.49\textwidth}
    \centering
    \adjustbox{valign=c}{\includegraphics[width=\linewidth]{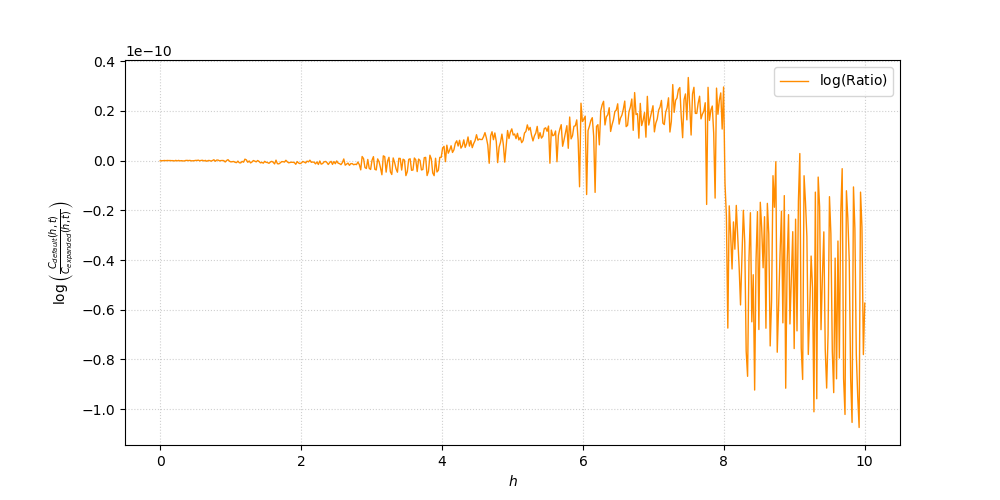}}
    \caption{ }
  \end{subfigure}
  \begin{subfigure}{0.49\textwidth}
    \centering
    \adjustbox{valign=c}{\includegraphics[width=\linewidth]{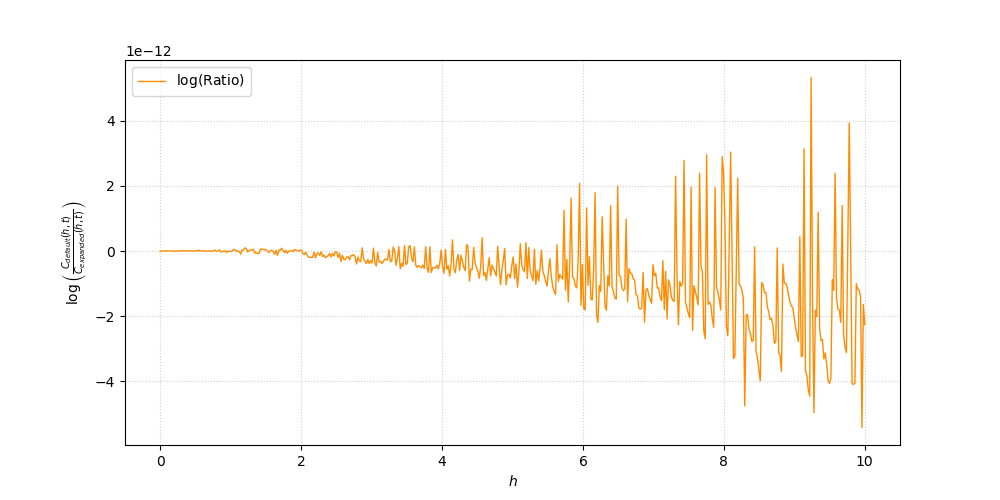}}
    \caption{ }
  \end{subfigure}\\
  \begin{subfigure}{0.49\textwidth}
    \centering
    \adjustbox{valign=c}{\includegraphics[width=\linewidth]{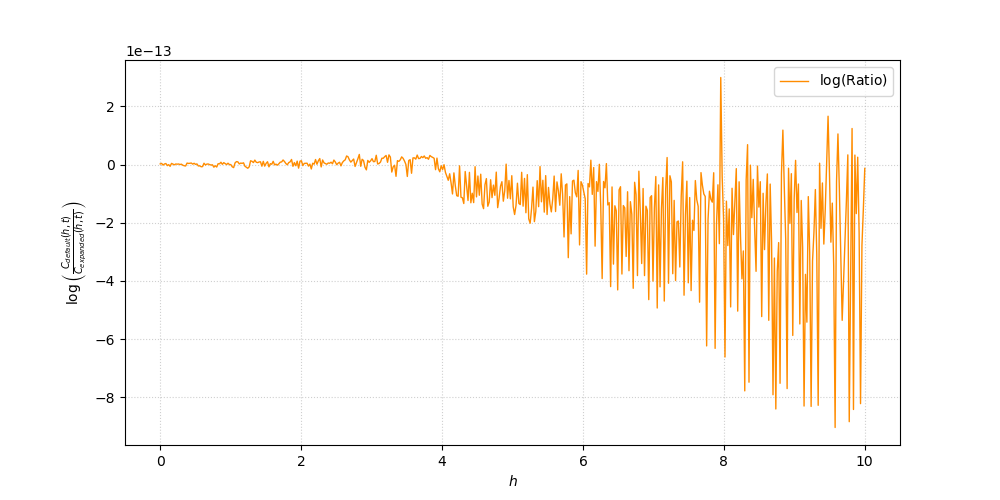}}
    \caption{ }
  \end{subfigure}
  \begin{subfigure}{0.49\textwidth}
    \centering
    \adjustbox{valign=c}{\includegraphics[width=\linewidth]{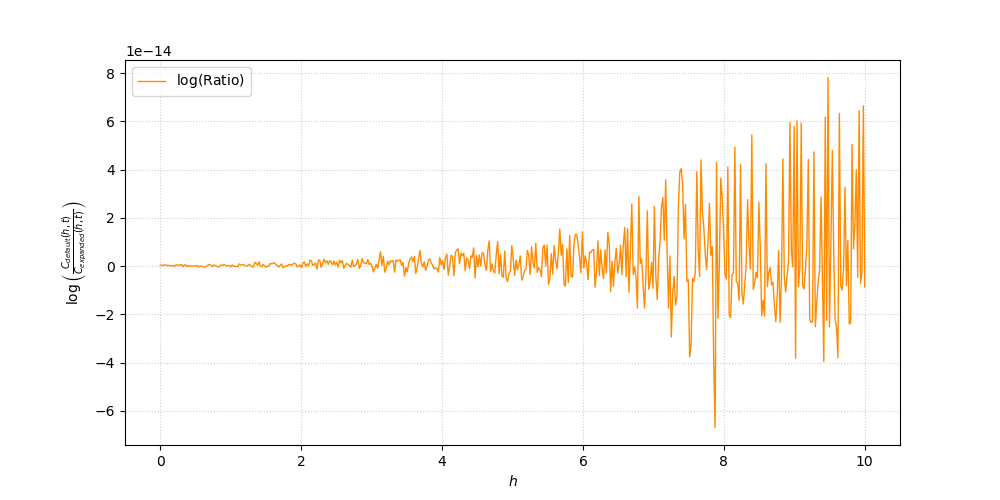}}
    \caption{ }
   \end{subfigure}
  \caption{The logarithm of the ratio between
 the direct evaluation of the Black-Scholes formula \eqref{eq:BSpriceCompact} (denoted here $C_{default}(h, t)$) 
 and the Mills-ratio-based implementation \eqref{eq:BSpriceMills}(denoted $C_{expanded}(h, t)$), 
 with varying $h$ and fixed $t$.
 (A) $t=0.001$, (B) $t=0.01$, (C) $t=0.1$ and (D) $t=1$. 
 The magnitudes of the oscillations caused by subtractive cancellation decay as $t$ increases.
 }
  \label{fig:TaylorOscillations_t_fixed}
\end{figure}
	\begin{figure}[h!]
  \centering
  \begin{subfigure}{0.49\textwidth}
    \centering
    \adjustbox{valign=c}{\includegraphics[width=\linewidth]{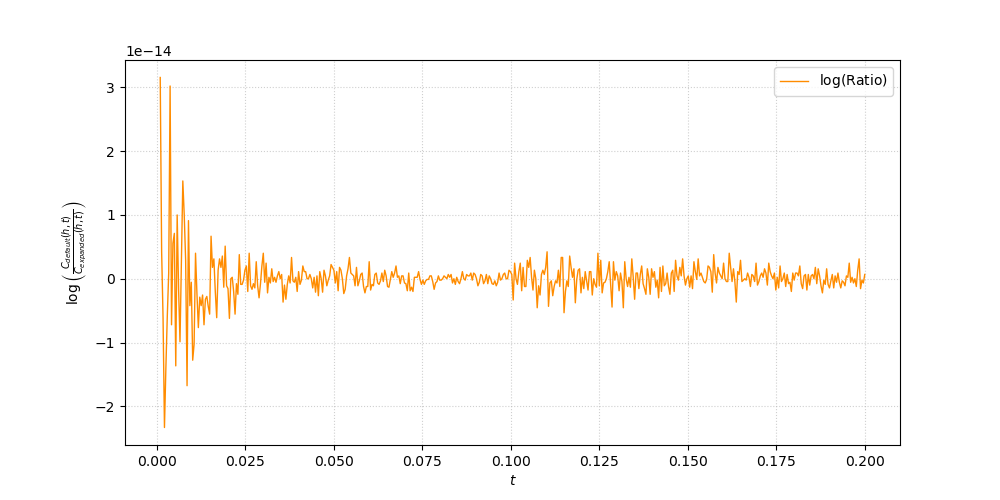}}
    \caption{ }
  \end{subfigure}
  \begin{subfigure}{0.49\textwidth}
    \centering
    \adjustbox{valign=c}{\includegraphics[width=\linewidth]{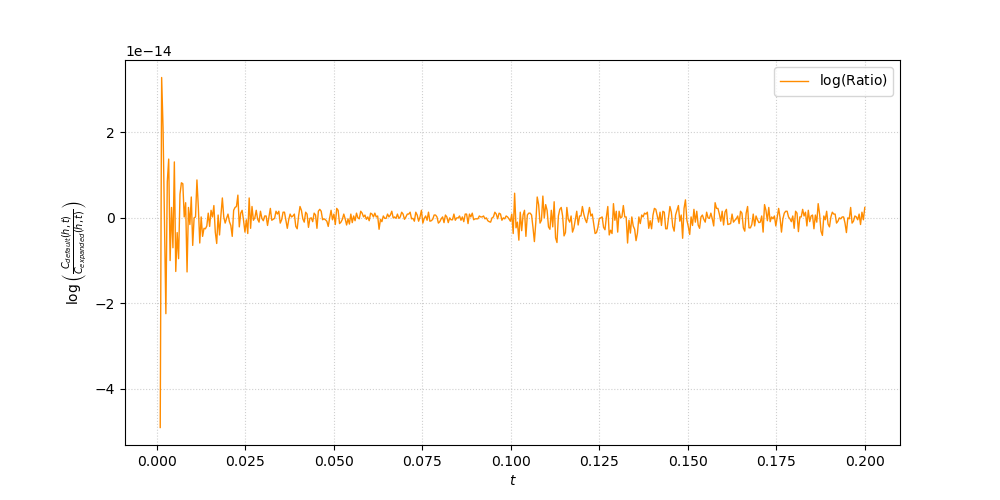}}
    \caption{ }
  \end{subfigure}\\
  \begin{subfigure}{0.49\textwidth}
    \centering
    \adjustbox{valign=c}{\includegraphics[width=\linewidth]{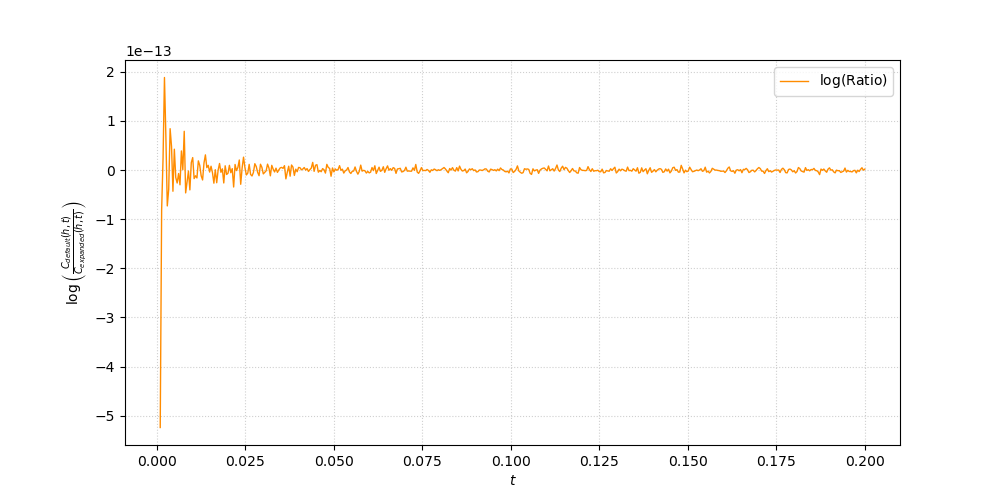}}
    \caption{ }
  \end{subfigure}
  \begin{subfigure}{0.49\textwidth}
    \centering
    \adjustbox{valign=c}{\includegraphics[width=\linewidth]{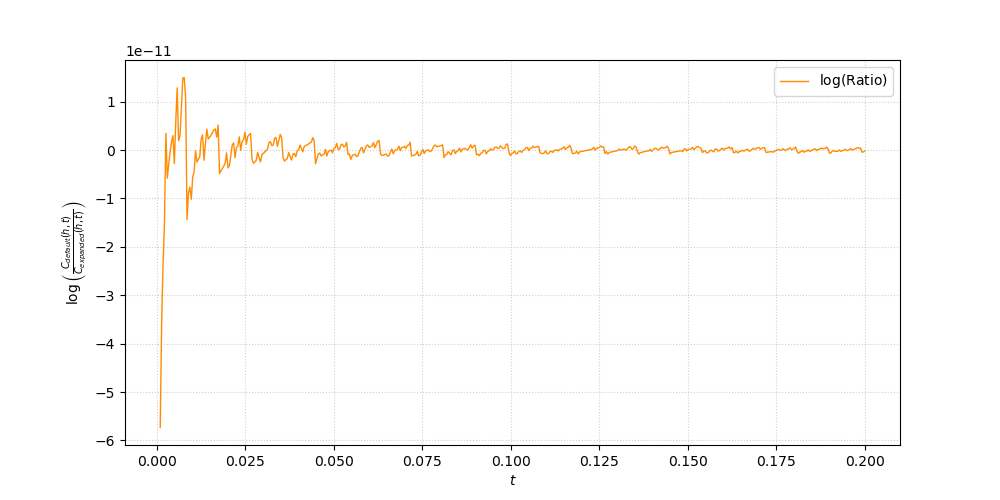}}
    \caption{ }
   \end{subfigure}
  \caption{
  The logarithm of the ratio between
 the direct evaluation of the Black-Scholes formula \eqref{eq:BSpriceCompact} (denoted here $C_{default}(h, t)$) 
 and the Mills-ratio-based implementation \eqref{eq:BSpriceMills}(denoted $C_{expanded}(h, t)$), 
 with varying $t$ and fixed $h$.
 (A) $h=0.01$, (B) $h=0.1$, (C) $h=1$ and (D) $h=10$. 
 The magnitudes of the oscillations decay as $t$ increases. Larger values of $h$ lead to an earlier disappearance of the oscillations.}
  \label{fig:TaylorOscillations_h_fixed}
\end{figure}
	\section{Householder root-finding method for the implied volatility}\label{ap:Householder}
	As discussed in Section~\ref{sec:NumericalExp}, the neural-network architectures developed in this paper are used both as stand-alone approximators of the implied volatility and as highly accurate initial guesses for a subsequent iterative refinement procedure. The refinement step is based on the third-order Householder method \cite{Householder}, also known as the cubic Householder iteration. This method generalises the classical Newton method (order $1$) and Halley's method (order $2$).
	
	Given a sufficiently smooth objective function $g$, the goal is to determine a root $B^\ast$ satisfying $g(B^\ast) = 0$. The third–order Householder iteration for solving such an equation can be compactly written as
	\begin{equation}\label{eq:HouseholderCompact}
	B_{n+1}
	= B_n
	+ \nu \,
	\frac{ 1 +  \frac{1}{2} \nu h_{2} }
	{ 1 + \nu\left( h_{2} + \frac{1}{6}\nu h_{3} \right) },	
	\end{equation}
	where we use the shorthand notation
	\[
	\nu := -\frac{g}{g'}(B_n), \qquad
	h_{2} := \frac{g''}{g'}(B_n), \qquad
	h_{3} := \frac{g'''}{g'}(B_n).
	\]
	
	The purpose of this appendix is to derive explicit expressions for these quantities in the context of Black-Scholes implied-volatility inversion and to explain how different objective functions are employed in the asymptotic regimes identified in Section~\ref{sec:BSAnalysis}. But before we carry on, let us first introduce some very useful notation. 
	
	Throughout this appendix, all functions are viewed as functions of the total volatility variable $B>0$, and the dependence on $B$ is omitted whenever no ambiguity arises. More explicitly, given a function $f$ depending on $B$, we will write $f$ instead of $f(B)$. If the function $f$ is smooth, and for all $i,j\in \Nbb$, we write $f_{i,j}:=\frac{f^{(i)}}{f^{(j)}}$ (assuming $f^{(j)}(B)\neq 0$.) 
	
	For a fixed value of the log-moneyness $A\geq 0$, we denote $C=C_{\mathrm{BS}}(A,B)$. As in the previous appendix, we use the variables $h=\frac{A}{B}$ and $t=\frac{B}{2}$.
	
	Recall that the normalised Black-Scholes price may be written as
	\[
			C = \Phi(u) - e^A \, \Phi(v), \qquad\text{with }
			u = -\frac{A}{B} + \frac{B}{2}, \quad
			v = -\frac{A}{B} - \frac{B}{2}.
	\]
	Since $v^2-u^2=2A$, we have $\varphi(u)=e^A \varphi(v)$. Then, with $u$ and $v$ seen as functions of $B$, one trivially has
	\[
	C'=(u'-v')\varphi(u)=\varphi(u).
	\]
	Hence
	\[
	C''=u'\varphi'(u)=-uu'\varphi(u)
	\quad \text{and} \quad
	C'''=\varphi(u) ((uu')^2-uu''-(u')^2).
	\]
	Consequently
	\begin{equation}\label{eq:C_21C_31}
	C_{2,1}=\frac{C''}{C'}=-uu'
	\quad \text{and} \quad
	C_{3,1}=\frac{C'''}{C'}=(uu')^2-uu''-(u')^2.	
	\end{equation}
	Moreover,
	\[
	u'=\frac{A}{B^2} + \frac{1}{2}
	\quad \text{and} \quad
	u''=-2\frac{A}{B^3}.
	\]
	Thus, for given $A>0$ and a target price $0<C_0<1$ in the central region (which corresponds to a relatively moderate value of the implied volatility $B$), it is numerically safe to directly use the natural objective function $g_0: B\mapsto C_{\mathrm{BS}}(A,B)-C_0$ with the recursive scheme \eqref{eq:HouseholderCompact} and the explicit formulae obtained above. More explicitly
		\[
	\nu = -\frac{C-C_0}{\varphi(u)}
		\quad ,\quad
	h_{2} =  -uu'
		\quad \text{and} \quad
	h_{3} =(uu')^2-uu''-(u')^2.
	\]
	
	 When $B\to 0^+$, both $u$ and $v$ tend to $-\infty$ and all of the quantities $C, C', C'', C'''$ converge to $0$. In this regime, the direct residual $g_0$ becomes numerically unstable. Following \cite{Jackel}, we therefore replace it by the alternative residual
	 \[ g_\lrm(B) = \frac{1}{\log C(B)} - \frac{1}{\log C_0}. \] 
	 This residual is employed whenever the given price satisfies
	 \[ C_0 \le C_\lrm= C_{\mathrm{BS}}(A,B_\lrm(A)). \]
	For a given price $C_0\leq C_\lrm$, a direct computation gives
		\begin{align*}
			g_\lrm' &= -\frac{C'}{C} \frac{1}{(\log C)^2}, \\
			g_\lrm'' &= -\frac{C''}{C} \frac{1}{(\log C)^2} + \frac{C'^2}{C^2} \frac{1+\frac{2}{\log C}}{(\log C)^2}, \\
			g_\lrm''' &= -\frac{C'''}{C} \frac{1}{(\log C)^2} 
		+ \frac{C'}{C} \left[  3 \frac{C''}{C} \frac{1+\frac{2}{\log C}}{(\log C)^2} -  2\frac{ C'^2}{C^2} \frac{1+\frac{3}{\log C} \,(1+\frac{1}{\log C})}{(\log C)^2} \right].
		\end{align*}
		Thus, for the objective function $g_\lrm$, and denoting $L=\log C$, we have
			\begin{align*}
			\nu &= \frac{(\log C_0 - L)}{\log C_0} \frac{C}{C'}L, \\
			h_2 &= \frac{C''}{C'} - \frac{C'}{C} \left(1 + \frac{2}{L}\right), \\
			h_3 &= \frac{C'''}{C'} 
					+ 2 \frac{C'^2}{C^2} \left(1 + \frac{3}{L} \left(1+\frac{1}{L}\right)\right)
					- 3\frac{C''}{C} \left(1+\frac{2}{L}\right).
		\end{align*}
		
	To evaluate these quantities efficiently, we use the expansions and notation introduced in Appendix~\ref{ap:AsymptoticBS}. We first note that
	\[
	C_{0,1}= \frac{C}{C'} = Z(h-t)-Z(h+t).
	\]
	Analysing again depending on the values of $h$ and $t$, we compute or approximate the above quantity using $\Delta Z(h,t)$ (introduced in the previous appendix and whose formulae are summarised in Subsection \ref{subsec:SummaryExpansions}). Accordingly, we have
	\[
			\nu = \frac{\log C_0 - L}{\log C_0} C_{0,1} L
			,\;\;
			h_2 = C_{2,1} - \frac{1 + \frac{2}{L}}{C_{0,1}}
			\text{ and }
			h_3 = C_{3,1} 
			+\frac{2 + \frac{6}{L} \,(1+\frac{1}{L})}{C_{0,1}^2}  - 3\left(1+\frac{2}{L}\right) \frac{C_{2,1}}{C_{0,1}},
		\]
	with $C_{2,1}$ and $C_{3,1}$ given by Equation \eqref{eq:C_21C_31} and $L$ and $C_{0,1}$ approximated by
	\[
	L\simeq\log (\Delta Z(h,t)) - \frac{1}{2}\ln (2\pi) - \frac{(h-t)^2}{2}
	\text{ and }
	C_{0,1}\simeq  \Delta Z(h,t).
	\]
	
	A similar difficulty arises as $B\to+\infty$. In this case, $u \to +\infty$, $v \to -\infty$, $C \to 1$ and the quantities $C', C''$ and $C'''$ converge to $0$. Consequently, the Householder root finding scheme based on the objective function $g_0$ becomes unstable. Similarly to the above, we introduce the alternative residual
		\[
		g_\urm(B) = \log \frac{1-C_0}{1-C(B)},
		\]
		to be used whenever the given price satisfies
		\[ C_0 \geq C_\urm=C_{\mathrm{BS}}(A,B_\urm(A)).\]
	In this case, we directly compute the successive derivatives
	\[
	g' = \frac{C'}{1-C}
	\; ,\;
	g'' =\frac{C''}{1-C} + \frac{(C')^2}{(1-C)^2}
	\; ,\;
	g''' =\frac{C'''}{1-C} + 3\frac{C'C''}{(1-C)^2}+ 2\frac{(C')^3}{(1-C)^3}.
	\]
	Hence
	\[
	\nu=-\log\!\bigg(\frac{1-C_0}{1-C}\bigg) \frac{1-C}{C'}
	\; ,\;
	h_2 = \frac{C''}{C'} + \frac{C'}{1-C}
	\; ,\;
	h_3 = \frac{C'''}{C'} + 3\frac{C''}{1-C}+ 2\frac{(C')^2}{(1-C)^2}.
	\]
	To obtain stable expressions, observe that
	\[
	1-C = 1-\Phi(u) + e^A \, \Phi(v) =\Phi(-u) + e^A \Phi(v).
	\]
	Therefore
	\[
	\widetilde{C}_{0,1} 
	= \frac{1-C}{C'}
	= \frac{\Phi(-u) + e^A \Phi(v)}{\varphi(u)}
	= Z(u) + Z(-v).
	\]
	Since both terms tend to zero for large $B$, direct evaluations of $Z(u)$ and $Z(-v)$ may underflow. We instead compute
	\[
	\log(1-C)=\log\left(\frac{1-C}{C'}\right)+\log(\varphi(u))
	= - \log(2)-\frac{u^2}{2}+\log\left( \mathrm{erfcx}\left( \frac{u}{\sqrt{2} } \right)+ \mathrm{erfcx}\left( \frac{-v}{\sqrt{2} } \right)\right).
	\]
	Defining $\widetilde L = \log(1-C)$, we obtain
	\[
	\nu=(\widetilde{L}-\log(1-C_0))\widetilde{C}_{0,1}
	,\;\;
	h_2 = C_{2,1}+\frac{1}{\widetilde{C}_{0,1}} 
	\text{ and }
	h_3 = C_{3,1} + 3\frac{C_{2,1}}{\widetilde{C}_{0,1}}+ 2\frac{1}{\widetilde{C}_{0,1}^2} ,
	\]
	with $C_{2,1}$ and $C_{3,1}$ given by Equation \eqref{eq:C_21C_31} and 
	\[
	\widetilde{L} = \log(1-C)= - \log(2)-\frac{u^2}{2}+\log\left( \mathrm{erfcx}\left( \frac{u}{\sqrt{2} } \right)+ \mathrm{erfcx}\left( \frac{-v}{\sqrt{2} } \right)\right),
	\]
	and 
	\[
	\widetilde{C}_{0,1} = \sqrt{\frac{\pi}{2}} \left( \mathrm{erfcx}\left( \frac{u}{\sqrt{2} } \right)+ \mathrm{erfcx}\left( \frac{-v}{\sqrt{2} } \right)\right).
	\]

\end{document}